\documentclass[preprint,12pt,3p]{elsarticle}

\usepackage{amsfonts}
\usepackage{fancyhdr}
\usepackage{comment}
\usepackage{mathtools}
\usepackage{times}
\usepackage{amsmath}
\usepackage{changepage}
\usepackage{graphicx}%
\usepackage[style=base]{caption}
\usepackage{subcaption}
\usepackage{wrapfig}
\usepackage{multirow}
\usepackage{tabu}
\usepackage{color,soul}
\usepackage{amssymb}
\usepackage{calrsfs}
\usepackage{amsthm}

\DeclareMathAlphabet{\pazocal}{OMS}{zplm}{m}{n}

\usepackage{amsmath,amssymb}

\newtheorem{theorem}{Theorem}
\newtheorem{definition}{Definition}

\journal{Mechanical Systems and Signal Processing}

\begin{document}
\begin{frontmatter}

\title{PhyMDAN: Physics-informed Knowledge Transfer between Buildings for Seismic Damage Diagnosis through Adversarial Learning}

\author[label1]{Susu Xu\corref{cor1}}
\address[label1]{Civil and Environmental Engineering, Carnegie Mellon University, Pittsburgh, PA 15213, USA}
\cortext[cor1]{Corresponding author}
\ead{susux@alumni.cmu.edu}

\address[label2]{Civil and Environmental Engineering, Stanford University, Stanford, CA 94305, USA}

\author[label2]{Hae Young Noh}
\ead{noh@stanford.edu}

\begin{abstract}
    Automated structural damage diagnosis after earthquakes is important for improving efficiency of disaster response and city rehabilitation. In conventional data-driven frameworks which use machine learning or statistical models, structural damage diagnosis models are often constructed using supervised learning. Supervised learning requires historical structural response data and corresponding damage states (i.e., labels) for each building to learn the building-specific damage diagnosis model. However, in post-earthquake scenarios, historical data with labels are often not available for many buildings in the affected area. This makes it difficult to construct a damage diagnosis model. Further, directly using the historical data from other buildings to construct a damage diagnosis model for the target building would lead to inaccurate results. This is because each building has unique physical properties and thus unique data distribution. 

To this end, we introduce a new framework, Physics-Informed Multi-source Domain Adversarial Networks (PhyMDAN), to transfer the model learned from other buildings to diagnose structural damage states in the target building without any labels. This framework is based on an adversarial domain adaptation approach that extracts domain-invariant feature representations of data from different buildings. The feature extraction function is trained in an adversarial way, which ensures that extracted feature distributions are robust to variations of structural properties. The feature extraction function is simultaneously jointly trained with damage prediction function to ensure extracted features being predictive for structural damage states. With extracted domain-invariant feature representations, data distributions become consistent across different buildings. We evaluate our framework on both numerical simulation and field data collected from multiple building structures. The results show up to $90.13\%$ damage detection accuracy and $84.47\%$ damage quantification accuracy on simulation data, and up to $100\%$ damage detection accuracy and $69.93\%$ 5-class damage quantification accuracy when transferring from numerical simulation data to real-world experimental data, which outperforms the state-of-the-art benchmark methods. 
\end{abstract}

\begin{keyword}
Structural damage diagnosis \sep  Domain Adaptation \sep Transfer Learning\sep Statistical signal processing
\end{keyword}

\end{frontmatter}
\section{Introduction}
\label{sec:intro}
Accurate and timely structural damage diagnosis is important to ensure human safety and assist the process of city reconstruction in post-earthquake scenarios~\cite{matsui2012structural}. After an earthquake, building damage diagnosis helps people select safe shelters and mitigates secondary injuries from damaged buildings. Diagnosing the damage severity of different buildings provides detailed information for decision-making, including the scale of city devastation and repair/demolishing plans of damaged structures, and thus facilitates city reconstruction~\cite{foltz2004estimating}.

% Earthquake-induced building damage diagnosis methods have two main approaches: physics-based and data-driven approaches. Conventional physics-based methods often require detailed knowledge of the building properties and explicit sophisticated nonlinear model of structures~\cite{tremblay1996seismic,uang1997seismic}. Obtaining the prior knowledge and estimating the physical model are very time-consuming and make the ``fast response" in the aftermath of an earthquake impossible. To overcome the problem, 
In recent years, many data-driven earthquake-induced building damage diagnosis methods based on structural vibration sensing have been developed and received much attention~\cite{young2011use,noh2012development}. By combining advanced machine learning techniques in a supervised fashion, these intelligent data-driven methods utilize historical structural response data and corresponding true damage state (label) to train building-specific models for structural damage diagnosis~\cite{hwang2018nonmodel,pakzad2009statistical,goulet2015data,xu2017information}. The data-driven methods can be extended to various application scenarios as they get rid of the requirements of prior knowledge about building structures. In this paper, we focus on two aspects of damage diagnosis: whether there exists damage (damage detection) and the severity of the damage (damage quantification). These methods are well-suited to modeling the non-linear properties of damage diagnosis, and perform well when there is massive historical data collected from the same building of interest.
%1) the same structure of interest, 2) under similar types of ground motions as what currently induces the damages. 

However, an extensive collection of historical data, especially with the true labels, is difficult and expensive, if not impossible, in real-world practices. %Firstly, it is very expensive and dangerous to collect the historical labels. 
\textcolor{black}{The ``true labels" refers to a collection of labels for data points (e.g., damage states in our scenario)}.
One of the common methods of collecting these labels is to conduct visual inspections by trained human experts, which is labour-intensive, time-consuming, and dangerous. 
\textcolor{black}{In recent years,  researchers have designed and instrumented different types of sensors to directly measure the structural damage, but these sensors are either expensive to deploy or have strict requirements of operation conditions. For example, story drift ratio (SDR), defined as the relative translational displacement between adjacent floors, is often employed to indicate structural damage states~\cite{young2011use,ghobarah2001performance,vamvatsikos2002incremental}. However, sensors for directly sensing SDR, such as a linear variable differential transducer (LVDT), laser scanning, and augmented reality visualization,  require with a wire diagonally strung across the structure of interest, involve expensive sensors, or need a human operator~\cite{hou2018monitoring, skolnik2010critical, skolnik2008instrumentation, dai2011photogrammetry, mccallen2017laser,jennings2014retrofit,pang2009simplified}. }
%It is costly to widely deploy these sensors on buildings in earthquake prone areas. %Besides, the frequency of similar types of earthquakes happening to the same building of interest is very low. 
Meanwhile, even though we are able to collect structural damage labels on some buildings, there are still a limited number of labeled data available due to the rarity of earthquakes. Limited quantity or lack of labels makes it difficult to train a damage diagnosis model sufficiently, which finally reduces damage diagnosis accuracy. 

Meanwhile, directly adopting the diagnosis models from other buildings to diagnosing the building of interest often results in low performance.  Many existing supervised learning methods are developed under the assumption that the data used for constructing a model (training data) have the same joint distribution of the inputs and labels as the data to be predicted based on the model (test data)~\cite{pan2010survey}. \textcolor{black}{However, in practices, as Figure~\ref{fig:transfer_idea} shows, different buildings often have distinct characteristics such as geometries, material properties, and foundation conditions. Their reactions to different earthquake excitation also vary significantly. Therefore, their data distributions are different from one another.} Directly adopting models learned from other buildings to diagnose new buildings violates the aforementioned assumption. \textcolor{black}{As the assumption is violated, the adopted models will fail to generalize to diagnose damage for different buildings, which significantly constrains the wide applications of data-driven damage diagnosis methods in post-earthquake scenarios.} \textcolor{black}{Traditionally, a straightforward way to address the inconsistent data distributions is to fine-tune an adopted/pre-trained model learnt from other buildings on the dataset collected from the building of interest~\cite{azimi2020structural,jang2019deep}. However, this method often requires large amounts of labeled data from the building of interest, which is not practical.}
%Especially, combining the historical data from multiple other buildings to train one model and applying it to the building of interest may make the performance worse since the distributions across these buildings are already quite distinct to fit into a single model. 

% \begin{figure}[htp!]
% \begin{center}
% \includegraphics[scale=0.6]{Figure/explanation.jpg}
% 	\caption{The data distribution (features, damages) of building 1 under Earthquake 1 does not equal to the distribution of building 2 under Earthquake 2 in real-world practices, which violates the underlying assumptions of most current supervised learning methods that data distribution is consistent across the training and test dataset.}
% 	\label{fig:explanation}
% 	\vspace{-0.5cm}
% \end{center}
% \end{figure}

To this end, the machine learning community has introduced domain adaptation techniques to address these knowledge transfer problems~\cite{pan2010survey, zhang2013domain,haeusser2017associative,pan2011domain,ghifary2016deep}. \textcolor{black}{In this paper, we denote one building structure as a single domain.} We also denote the buildings from which labeled data are collected and the knowledge are learned as ``Source Domain," and the new building of interest without any labeled historical data as ``Target Domain". \textcolor{black}{Domain adaptation is a subcategory of transfer learning, developed to adapt a collection of source domains for transferring knowledge to one or more target domains and improve the learning performance in target domains.}
%This knowledge transfer problem is called multiple source domain adaptation when having the historical data from multiple different buildings (multiple source domains) to learn. 
Previous domain adaptation studies fall into several distinct categories such as instance-based~\cite{zhang2013domain} and feature representation-based~\cite{ghifary2016deep,haeusser2017associative,pan2011domain}. \textcolor{black}{Instance-based approach re-weights/sub-samples the source samples to match the data distributions of the source and target domains. However, in practice, the distribution difference between the source and target domains is often too complex to be reduced through direct re-weighting/sub-sampling in the original feature space~\cite{weiss2016survey}. To this end, a more common method is feature representation-based domain adaptation. This approach aims to find an optimal projection from the original feature representation space to a new space where the data distributions of the source and target domain can be best matched.} In this new feature space with consistent training and test data distributions, we can directly adopt the model learned from source domains to diagnose the damage state of the target building. A successful feature representation-based domain adaptation ensures the extracted feature of different buildings to be ``domain-invariant" and ``discriminative".  \textit{\textbf{``Domain-invariant"}} features have consistent distributions across different domains. Meanwhile, the extracted feature representations need to be \textit{\textbf{``discriminative"}} for structural damages to ensure the damage diagnosis accuracy on the source buildings. 

However, there are still challenges to knowledge transfer from source buildings to the target building for earthquake-induced structural damage diagnosis.
%\emph{First}, different buildings may suffer from different earthquakes with fast-changing dynamics and noises, which would influence the structural responses and damage patterns of buildings. When transferring the knowledge across buildings, it is important to differentiate the data distribution changes induced by the variations of earthquakes from that of buildings. Otherwise, the influence of earthquakes will mislead the modeling However, the influences of earthquakes are coupled in building responses in a complicate way and hard to be distinguished. Conventional features based on system identification, such as frequency domain decomposition,  do not contain sufficient information to the non-stationary earthquake excitations. 
\emph{First}, there exists a trade-off between domain-invariance and discriminativeness of domain projection. On the one hand, the damage-predictive information is coupled with earthquake-induced structural responses and environmental noises, which vary with buildings, meaning that they are domain-variant. Simply forcing the domain-invariance of feature representations would easily eliminate the damage-predictive information contained in the raw untransformed data, which makes the extracted features poor predictors for structural damage. On the other hand, if we focus on extracting discriminative features regardless of underlying distribution changes across different buildings, it would be difficult to ensure the domain-invariance of extracted features. The most extreme case is to directly use raw data, which contains as much information as possible for damage prediction but makes it difficult to transfer knowledge across buildings. Therefore, the trade-offs between domain-invariance and discriminativeness in the earthquake-induced building damage diagnosis problem need to be investigated and resolved.

\emph{Second}, the changes of distributions between source buildings and the target building are too complex to be explicitly modeled, which makes it difficult to best extract domain-invariant feature representations. The joint distribution of the inputs and labels of each building is related to the building's structural properties, i.e., the functional mapping from building properties to the structural damages given earthquake excitation. This distribution varies with building according to different soil types, foundations, complex structural and non-structural components, and other influence factors.
%, which depend on complicate structural and non-structural components. This makes it difficult to model the distribution changes and extract feature representations which retain domain-invariance and discriminativeness simultaneously. 
%During the earthquake, the soil-structure interaction system is a fast-changing highly dynamic time-variant system. The soil types, foundations, complicate structural and non-structural components all have impacts on building damage patterns. 
With limited prior knowledge of the influences, it is difficult to directly model the damage patterns of each building accurately.

\emph{Last but not least}, the data from different source buildings may induce different levels of biases to the learned model, causing challenges to integrating the knowledge from multiple source buildings. Due to the aforementioned problem of costly label collection, each source building usually has a limited amount of data, and thus provides limited information about earthquake-induced structural damage patterns. One possible solution is to integrate and transfer the knowledge learned from multiple source buildings, 
%The data from multiple different buildings provides richer information about structural damage patterns compared to a single building. The knowledge transfer from multiple source buildings 
which is referred to ``multiple source domain adaptation". Most conventional multiple source domain adaptation methods assume that all of the source domains have the equal importance to provide equally sufficient information to learn domain-invariant features~\cite{hoffman2017multiple}. \textcolor{black}{However, in real-world practices, different source buildings may have heterogeneous physical properties (heterogeneous source buildings), resulting in distinct data distributions across source domains. Meanwhile, the size of source domains also varies, inducing distinct levels of bias in the statistical estimations of source data distributions.} Treating all source buildings with the same importance ignores the source domain differences, and thus reduce the performance of knowledge transfer. \textcolor{black}{Some previous studies address this problem through measuring how closely related each source data is to target data. These methods either require labeled target data or are based on strict smoothness assumptions of the conditional and marginal data distributions~\cite{sun2011two, chattopadhyay2012multisource,seah2012combating, ge2014handling}. However, since the geometry of a high-dimensional data manifold is complex to estimate, it is inherently difficult to find an effective and robust measurement of relatedness from limited data ~\cite{weiss2016survey}.} %This is because that the discriminativeness of extracted features on the target building would be sacrificed to ensure the features' domain-invariance on those distinct source buildings, which finally reduces the performance of knowledge transfer. 

To address these challenges, we introduce a new Physics-Informed Multi-source Domain Adversarial Network, (PhyMDAN), which transfers the knowledge learned from multiple different source buildings to help diagnose the story-wise damage conditions of the target building structure without any labeled data. This framework integrates a feature extractor, a domain discriminator, and a damage predictor with deep neural network architectures. The feature extractor aims to extract feature representations which are domain-invariant and discriminative. The domain discriminator regularizes the domain-invariance of extracted features. The damage predictor ensures the discriminativeness of extracted features. To find an optimal trade-off between domain-invariance and discriminativeness, our framework jointly optimizes the feature extractor, domain discriminator and damage predictor. The joint optimization is formulated as a minimax problem. By finding the saddle point, the algorithm converges to the near-optimal trade-off between domain-invariance and discriminativeness.
%Then we utilize the convolutional neural nets, which is a powerful deep neural networks to understand the hierarchical data patterns and assemble complex information from different frequency bands~\cite{krizhevsky2012imagenet}. This new input features help the statistical model better understand the local correlations between ground motions and floor vibrations at different frequency bands. 
To learn the underlying distribution changes across buildings without any prior knowledge, we train the feature extractor and the domain discriminator in an adversarial way, where the feature extractor best learns the underlying domain-invariant representations such that the domain discriminator hardly distinguishes the differences between domains. The extracted domain-invariant features enable a better understanding of physical relationships between various building properties, earthquake excitation and damage distributions. To eliminate the biases introduced by less similar source buildings' distributions, we design new physics-guided weights based on similarities between buildings. The source buildings which have more similar physical properties to the target building have higher weights in training loss function. The new weight design ensures the algorithm to prioritize knowledge transfer from similar source buildings and reduce the biases and noises induced by other source buildings. The experiments show the new weights significantly improve the performance in real-world practices. 
%However in our scenarios, the importance of different source buildings depends on the similarity of physical properties between the source building and the target building.

This work has four main contributions:
\begin{itemize}
    \item[1.] \textcolor{black}{To the best of our knowledge, the introduced framework is the first multiple source domain adaptation framework for earthquake-induced building damage diagnosis without any labels of the target building.} This end-to-end framework integrates data augmentation, input feature extraction, and adversarial domain adaptation for damage diagnosis tasks including damage detection and damage quantification. This framework ensures both the discriminativeness and domain-invariance of the extracted features by jointly optimizing the feature extractor, domain discriminator and damage predictor.
    %This multiple source domain adaptation framework is flexible to integrate and transfer the knowledge learned from multiple source buildings to help diagnose the target building. 
    \item[2.] We introduce adversarial domain adaptation to learn and transfer the knowledge from \textcolor{black}{multiple  source buildings with heterogeneous data distributions}. The adversarial training scheme avoids the need for complex sampling or inference during the training process, which enables an efficient learning of underlying domain-invariant feature representations and is robust to complicated distribution changes.
    \item[3.] We design a new physics-guided weight in the loss function based on the physical similarities of buildings. The new weights help eliminate the biases introduced by those source buildings with less physical similarities to the target building. We prove that our new physics-guided loss provides a tighter upper bound for the damage prediction risk on the target domain compared to the general loss without combining physical knowledge.
    \item[4.] We characterize the performance of our framework using both numerical simulation data and real-world experimental data, including 5 different buildings subjected to 40 earthquakes for simulation and a real-world 4-story building. Especially, we implement the  knowledge transfer across  simulation data, and from  simulation data to real-world experimental building damage diagnosis. The results show that our framework outperforms by other methods in both tasks.
    %Our method achieve upto $90.13\%$ damage detection accuracy and $84.66\%$ damage quantification accuracy on simulation data. We also successfully transfer the knowledge learned from the simulation data to the real-world building with $100\%$ damage detection accuracy and $69.93\%$ damage quantification accuracy, which outperforms by the state-of-the-arts frameworks.

\end{itemize}

\begin{figure}[ht!]
    \centering 
    \includegraphics*[scale=0.085]{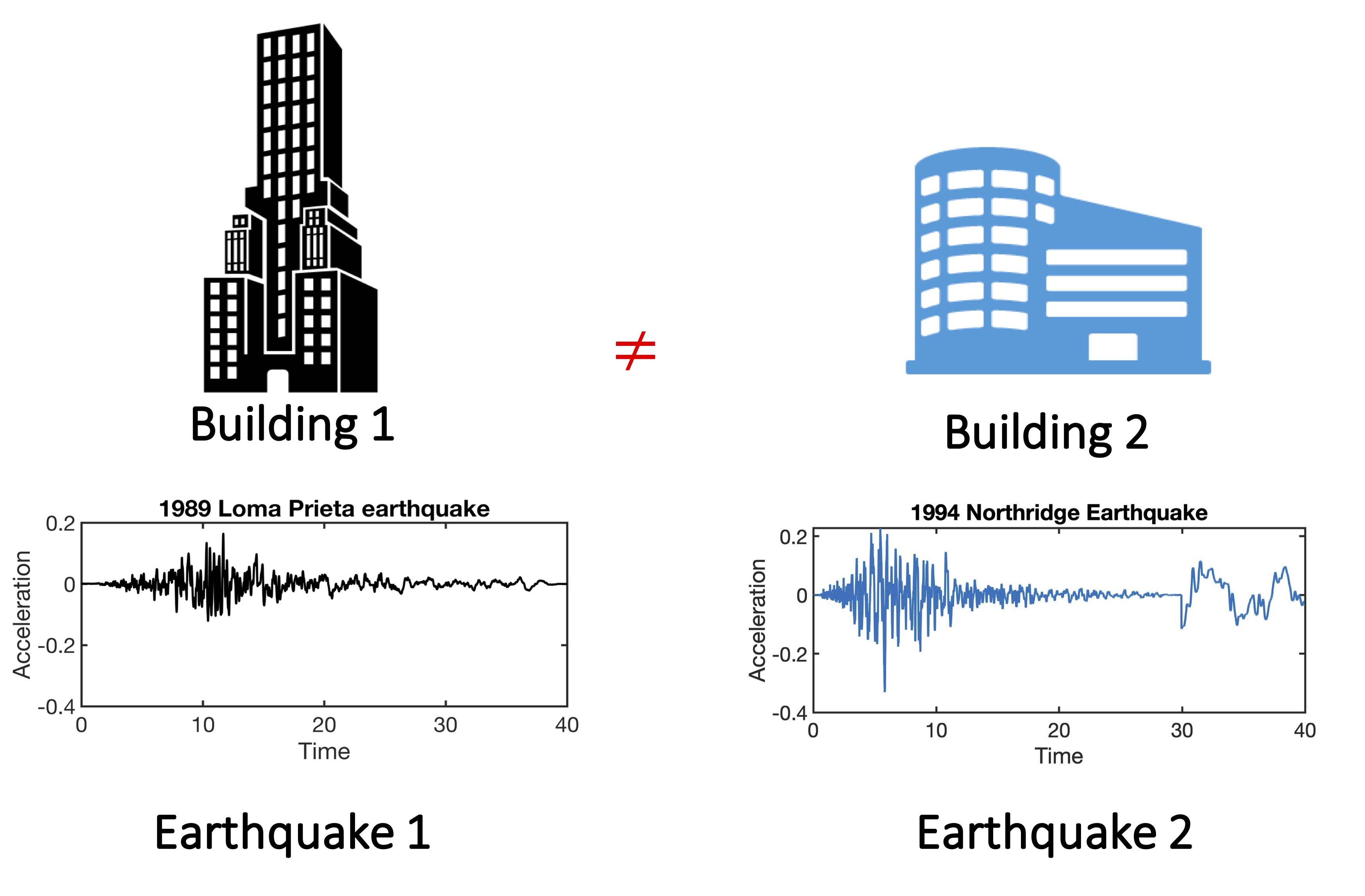}
    \caption{The data distribution (features, damages) of building 1 under Earthquake 1 does not equal to the data distribution of building 2 under Earthquake 2 in real-world practices, which violates the underlying assumptions of most current supervised learning methods that the data distribution is consistent across  training and test data.  \label{fig:transfer_idea}}
\end{figure} 

The remainder of this paper is organized as follows. Section~\ref{sec:relatedwork} discusses related work about the domain adaptation and its applications.
Section~\ref{sec:challenge} analyzes the domain adaptation challenges in earthquake-induced building damage diagnosis scenarios. Section~\ref{sec:framework} describes our knowledge transfer framework for building damage diagnosis, including the problem definition, model architectures, loss function design, and adversarial domain adaptation training scheme. Section~\ref{sec:eval} evaluates our knowledge transfer framework on knowledge transfer between numerical simulation data as well as from simulation data to experimental data. Section~\ref{sec:con} concludes the work and gives further discussions.

\section{Related Work }\label{sec:relatedwork}
% \subsection{Transfer Function for Building Damage Diagnosis}
% \subsection{Domain Adaptation}
Earthquake-induced building damage diagnosis is a challenging problem. 
In recent years, wide deployment of vibration sensing systems in buildings has provided rich building responses during earthquakes and enabled the applications of data-driven approaches for earthquake-induced building damage diagnosis~\cite{pakzad2009statistical,goulet2015data, lignos2008shaking,ji2010seismic}. Statistical models or machine learning techniques, such as multivariate linear regression~\cite{hwang2018nonmodel}, support vector machine~\cite{gui2017data}, kernel regression~\cite{young2011use}, deep autoencoder~\cite{pathirage2018structural}, and deep convolutional neural networks~\cite{abdeljaber2017real}, are utilized to extract damage indicators from structural responses during earthquakes, either in frequency domain or time domain, and then estimate the building damages~\cite{young2011use, noh2012development,  hwang2018nonmodel, hwang2017earthquake}. However, in practice, many buildings often have very limited or even no labels available for the collected structural response data, which makes it difficult to train the damage prediction model for the target building. Directly adopting the model learned from other buildings will result in low prediction performance due to inconsistent distributions of training and test data, which also constrains the utilization of valuable labeled datasets from other buildings.
%One possible solution is to combine the rich labeled datasets from other different buildings for model training. But this would violate a common underlying assumption shared by many of these machine learning techniques: the training and test data have the consistent feature space and the same data distribution. The violation of this assumption would significantly reduce the performance of statistical models or machine learning techniques. It also constrains the sufficient utilization of different labeled building datasets collected from previous earthquake events.

In machine learning communities, researchers introduced domain adaptation for knowledge transfer between inconsistent training and testing data distributions. 
Domain adaptation is one of classical transfer learning methods~\cite{weiss2016survey}. %Generally, transfer learning can be divided into 3 categories: inductive transfer learning, transductive transfer learning and unsupervised transfer learning~\cite{pan2010survey}. Given the source domain and the target domain as defined in Section~\ref{sec:intro}, inductive domain adaptation means that learning tasks in the source domain are different from the tasks in the target domain~\cite{raykar2008bayesian,yao2010boosting}. 
In domain adaptation, data distributions are different for the source and target domain, but the learning tasks, e.g., damage diagnosis, are the same across the source domains and target domain~\cite{daume2006domain, zadrozny2004learning, zhao2019learning, yu2018ilpc}. Especially, unsupervised domain adaptation focuses more on the unsupervised learning tasks in the target domain~\cite{dai2008self, wang2008transferred}, which means that the target domain data has no labels. In practice, most structural damage diagnosis tasks using other buildings' data are essentially addressing the domain adaptation problem, specifically, the unsupervised domain adaptation problem. This is because 1) the feature spaces are the same, although the distributions of features may be different. For example, we have the same sensing modality, e.g., floor vibration signals, for both source and target domains. 2) We have the same damage prediction tasks for different buildings since the definition of building damage states are consistent in source and target domains. 3) There is often no label available from the target building during earthquake excitation. \textcolor{black}{It is possible that sometimes undamaged labels will be available for the target structure. These undamaged labels may be collected under ambient excitation or previous earthquake excitations. However, ambient excitation has quite different vibration patterns compared to earthquake excitation, for example, ambient excitation is mostly stationary measurements while earthquake excitation is non-stationary~\cite{young2011use}. The target structures data distribution under ambient excitation will be very distinct from that under earthquake excitation. The target structure under weak ambient excitation should be treated as a new source domain. Therefore the problem still remains as an unsupervised domain adaptation problem. It is also possible, though rarely happening, that a limited number of labels are collected from previous earthquake excitations. In this case, this problem will become a semi-supervised domain adaptation problem~\cite{saito2019semi}. Our proposed framework can be easily generalized to semi-supervised domain adaptation scenarios, which will be discussed in Section~\ref{sec:adv_frame}.}

Traditional domain adaptation includes 2 types of approaches: instance-based and feature representation-based. \textcolor{black}{Instance-based learning assumes that certain parts of data from the source domains can be reused for the target domain. During the learning process, instance-based methods match the joint distribution $P(X,Y)$ between the source and target domain by re-weighting the labeled instances $(X,Y)$ from the source domains, where $X$ is the input feature and $Y$ is the label. To find the optimal re-weighting schemes, these approaches often need labeled data from the target domain~\cite{weiss2016survey,yao2010boosting,pan2010survey}}. Feature representation-based methods focus on learning a representation of the input $X$ in a reproducing kernel Hilbert space (RKHS) in which the feature distributions of different domains are close to each other (\textbf{domain-invariant}). Our solution falls into the feature representation-based class, which aims to first extract domain-invariant feature representations and then build an optimal damage classifier based on the extracted feature representations. Especially, we need to learn domain-invariant representations across \emph{multiple} source domains and the target domain.

%The main goal of our knowledge transfer is to extract feature representations which are both domain-invariant and discriminative with respect to the learning task. 
In recent years, feature representation-based domain adaptation has been widely studied in different areas, such as image recognition~\cite{pan2011domain,pan2010survey}, natural language processing~\cite{koehn2007experiments}, sentiment analysis~\cite{blitzer2007biographies} and so on. Classic methods are developed and applied in these applications, including transfer kernel learning~\cite{long2014domain} (TKL), geodesic flow kernel~\cite{gong2012geodesic} (GFK), joint distribution adaptation~\cite{long2013transfer} (JDA), and transfer component analysis~\cite{pan2010domain, mirshekari2020step} (TCA). 
However, these methods mainly focus on learning shallow feature representations, which constrain their abilities to understand the complex distribution changes as well as the knowledge transfer performance. Meanwhile, these methods often separately conduct domain adaptation and feature learning, which makes it difficult to ensure the extracted features to be simultaneously discriminative and domain-invariant. \textcolor{black}{Recently, researchers also applied domain adaptation methods for knowledge transfer across different civil  structures, including echo pattern analysis~\cite{ye2017robust}, structural health monitoring~\cite{gardner2020application}, and footstep-induced floor vibration~\cite{mirshekari2020step}. However, most of these methods apply existing classical single-source-single-target domain adaptation methods, which cannot handle multiple source domains and complex distribution shifting across domains. }

%In recent years, deep learning models are embedded to conduct more complex transformation to extract feature representations~\cite{glorot2011domain,donahue2014decaf,yosinski2014transferable}. Deep neural networks have stronger expressive power to extract robust domain-invariant feature representations disentangling underlying exploratory factors of variations and hierarchically combining features according to their correlation with invariant factors. However, 

To ensure the discriminativeness and domain-invariance of the extracted feature representations, researchers developed adversarial frameworks to embed domain adaptation into the process of feature representation extraction. In the adversarial frameworks, a domain discriminator block is added to distinguish between the samples from a source domain and the target domain, which provides strong regularization to encourage the domain-invariance of the extracted feature representations ~\cite{ganin2014unsupervised,ganin2016domain, tzeng2017adversarial, tzeng2015simultaneous}. Deep domain confusion loss has been introduced to maximize the similarity between the data distributions of source and target domains~\cite{tzeng2015simultaneous}. Gradient reversal algorithm has been designed to achieve adversarial training by reversing the gradient of the domain discriminator during the back-propagation~\cite{ganin2014unsupervised}. \textcolor{black}{Adversarial discriminative domain adaptation alternatively learns the feature representations and trains the domain discriminator~\cite{tzeng2017adversarial}. This has been applied in bearing fault diagnosis~\cite{wang2020triplet}, but these applications focus on single-source-single-target problems.} When there are multiple source domains, the distribution changes become more complicated. Naive applications of those single-source-single-target methods would result in sub-optimal solutions.

Some existing multiple source domain adaptation approaches are mostly based on fixed feature representation learning and can not utilize effective deep neural network models~\cite{gan2016learning,hoffman2012discovering}. Zhao et al. introduced a deep adversarial domain adaptation method for multiple source domain adaptation~\cite{zhao2018adversarial}. This method combines multiple domain discriminators to extract domain-invariant features across multiple source domains and a target domain~\cite{zhao2018adversarial}. 
%New loss function is designed to ensure extracted features' discriminativeness for damage prediction on multiple source domains. 
This method treats all domains with the same importance for domain-invariant feature representation learning. However, it is difficult to ensure that each domain contains equally important information for the target tasks in real-world practices. \textcolor{black}{To better integrate the information from different source domains, previous studies determine the weights through measuring how closely related each source data is to target data. But these methods either require labeled data on the target domain~\cite{ge2014handling, chattopadhyay2012multisource}, or have strict smoothness assumptions of the conditional and marginal data distributions~\cite{sun2011two, seah2012combating}. These assumptions are not practical in high-dimensional complex data manifolds. Measurement of the relatedness based on inaccurate data manifold estimation is not robust nor reliable~\cite{weiss2016survey}. In the meantime, we do not have available labeled target data to find the optimal weights in the projected space. Therefore in this work, we design the physics-informed weights based on prior physical knowledge to approximate the differences between source domains and better integrate the information for knowledge transfer. Note that our method is different from instance-based method, as we re-weight the source domain (i.e., the entire data distribution of individual source domains) instead of individual instances. Besides, we do not require labeled target data to find the optimal weights, but we introduce physics-informed weights based on our prior physical knowledge. }

\section{Data Distribution Shifting Challenges for Post-earthquake Building Damage Diagnosis}\label{sec:challenge}
In this section, we first describe the shifting of joint distribution of the inputs and labels, which is a common challenge in domain adaptation for post-earthquake building damage diagnosis scenarios. Then we characterize the data distribution changes with different earthquakes and building types. 

Domain adaptation is important for enabling wide applications of data-driven post-earthquake building damage diagnosis in data-constrained scenarios. In the general framework of data-driven approaches, structural responses, e.g., the structural vibration signals, are collected as input $x$, the respective structural damage states are label $y$, and a relation between $x$ and $y$ are defined as a function $\pazocal{F}: x \rightarrow y$ in a discriminative way or a joint distribution $P(x, y)$ in a generative way~\cite{raina2004classification}. A common underlying assumption for general supervised learning methods is that the marginal distributions $P(x)$ and $P(y)$ and the joint distribution $P(x,y)$ are consistent between the training dataset and the test dataset.  

\begin{figure}[th!]
\begin{center}
	\begin{subfigure}{0.485\textwidth}
		\includegraphics[scale=0.25]{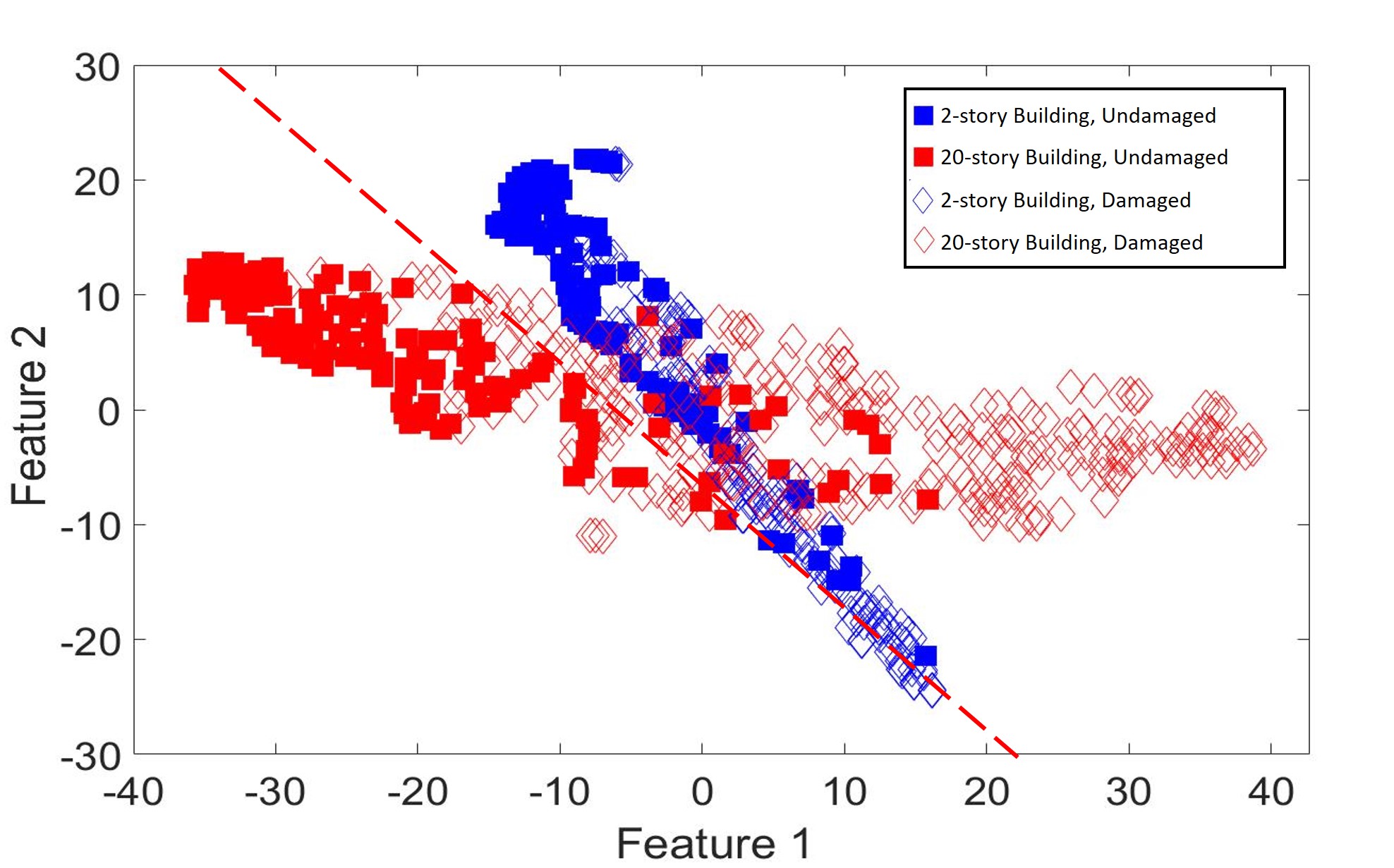}
		\caption{}
		\label{fig:1a}
	\end{subfigure}
	\begin{subfigure}{0.485\textwidth}
		\includegraphics[scale=0.26]{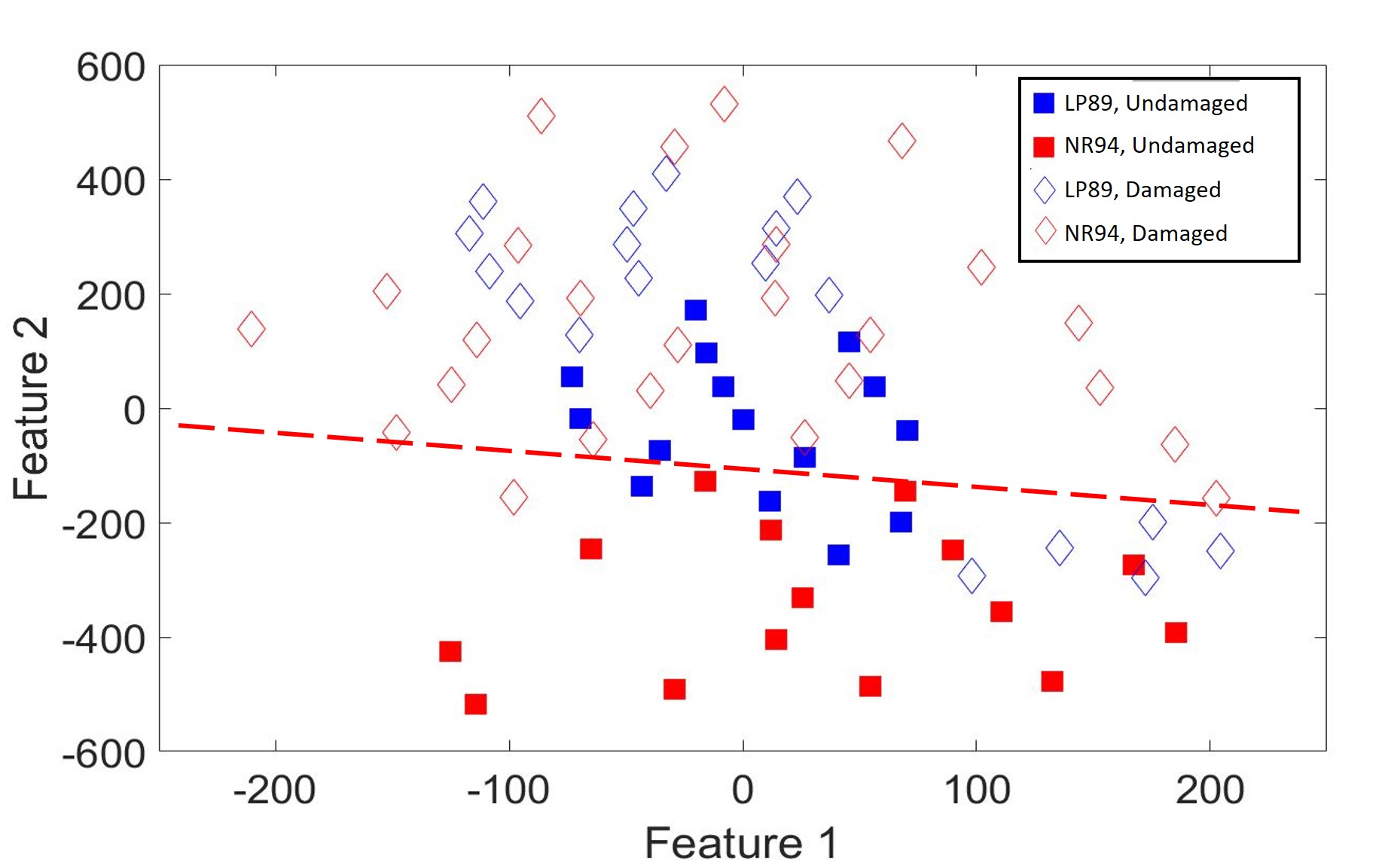}
		\caption{}
		\label{fig:1b}
	\end{subfigure}
	\caption{ The 2D tSNE visualizes the significant changes of data distributions of different buildings under the same earthquake and of the same building under different earthquakes. The red points represent data sampled from the source domain and the blue points represent data sampled from the target domain. The red dotted line shows the classification decision boundary for source domain (red points). Diamond indicates damaged and square indicates undamaged. (a) shows the difference between data distributions of a 20-story building and a 2-story building. (b) shows the difference between data distribution of 1994 Northridge earthquake and 1989 Loma Prieta earthquake. Both figures show that directly applying the model trained on source domain dataset (red) to diagnose structural damage on target domain (blue) results in low accuracy.}
\end{center}
\end{figure}

However, data distribution shifting commonly exists in model transfer between different buildings.
Denote the building with labeled data which we would like to transfer the knowledge from as source domain, i.e., $s$. Denote the building without any labeled data that we need to diagnose as target domain, i.e., $t$. The structural response and damage states collected from the source domain are $X^s$ and $y^s$, respectively, and those from the target domain are $X^t$ and $y^t$, respectively. The data distribution in the source domain does not equal to the data distribution in the target domain, i.e., $P(X^s, y^s)\neq P(X^t,y^t)$. Figure~\ref{fig:1a} presents the 2 dimensional t-Distributed Stochastic Neighbor Embedding (t-SNE) visualization~\cite{maaten2008visualizing} of structural vibration signals collected from a 2-story building (blue) and a 20-story building (red). The red dotted line refers to the true decision boundary for damage detection of the 20-story building (red points). The performance of damage detection for the 2-story building drops dramatically if the decision boundary learned from the 20-story building is used. Meanwhile, the data distributions of the same building under different earthquake excitation are also different. Figure~\ref{fig:1b} shows the visualization of structural vibration signals collected from the 1989 Loma Prieta (LP) earthquake (blue) and the 1994 Nothridge (NR) earthquake (red). The red dotted line represents the damage decision boundary for buildings under NR earthquake. Directly adopting the decision boundary learned from NR earthquake results in inaccurate damage decision for LP earthquake. 
%Therefore, it is an imminent challenge to solve in order to improve the ubiquitous applications of data-driven building damage diagnosis. 

\subsection{Characterization of Distribution Changes with Earthquakes and Buildings}
The joint distributions of the inputs and labels change with different earthquakes and buildings. The physical properties of earthquakes, including ground motion intensities, length, and seismic waveform, significantly influence the structural damage patterns. The physical properties of buildings ( e.g. stiffness, damping, structured design, and non-structural components) are also closely related to the damage occurrence. In this section, we visualize how ground motion intensity, which is a typical earthquake characteristic, and building height, which is a typical building property, would influence the complex data distribution changes.

\subsubsection{Impact of Ground Motion Intensities}\label{sec:impact_gm}

The previous research works show that the influence of earthquake excitations on structural damage patterns is significant and difficult to accurately model~\cite{trifunac1975correlation,spence1992correlation}. 
%When earthquake happens, for example, in reinforced concrete structures, seismic wave causes combinations of repeated stress reversals and high stress excursion and thus different structural damages~\cite{park1985seismic}.
The process of damage initiation and progression is highly non-linear, time-variant, and involves complicated wave propagation inside structures, which makes it difficult to quantitatively and accurately model the impact of earthquakes. It has been shown that the correlation between damages and ground motion characteristics are not consistent based on the observations in buildings in Mexico City~\cite{spence1992correlation}. There also exist many studies exploring load-deformation model relating the engineering ground motion parameters to the structural damages. For example, the ground motion intensity is  related to hysteretic energy~\cite{park1985seismic}. Destructiveness of seismic ground motions is also related to the seismic duration, maximum absolute ground acceleration, and frequency content of the respective strong ground motion~\cite{araya1985earthquake}. These studies show that it is difficult to build an explicit transfer function to model the distribution changes induced by a variety of earthquakes. %According to the analysis in~\cite{uang1997seismic, bertero1992lessons}, due to complicate structural inelastic deformation, only seismic ground motion information is not sufficient to infer the structural damage states. 

\begin{figure}[th!]
\begin{center}
	\begin{subfigure}{0.48\textwidth}
		\centering
		\includegraphics[scale=0.19]{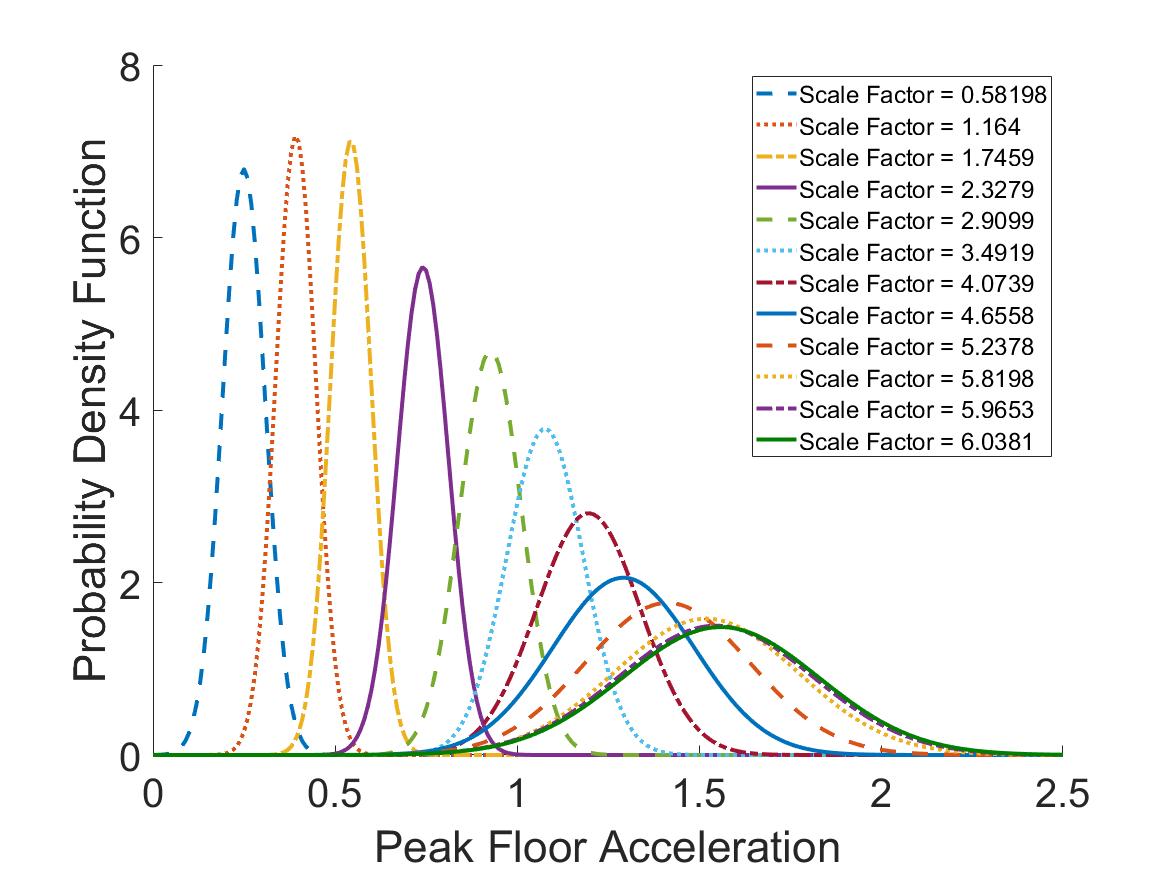}
		\caption{}
		\label{fig:pfa_sf}
	\end{subfigure}
	\begin{subfigure}{0.48\textwidth}
		\centering
		\includegraphics[scale=0.19]{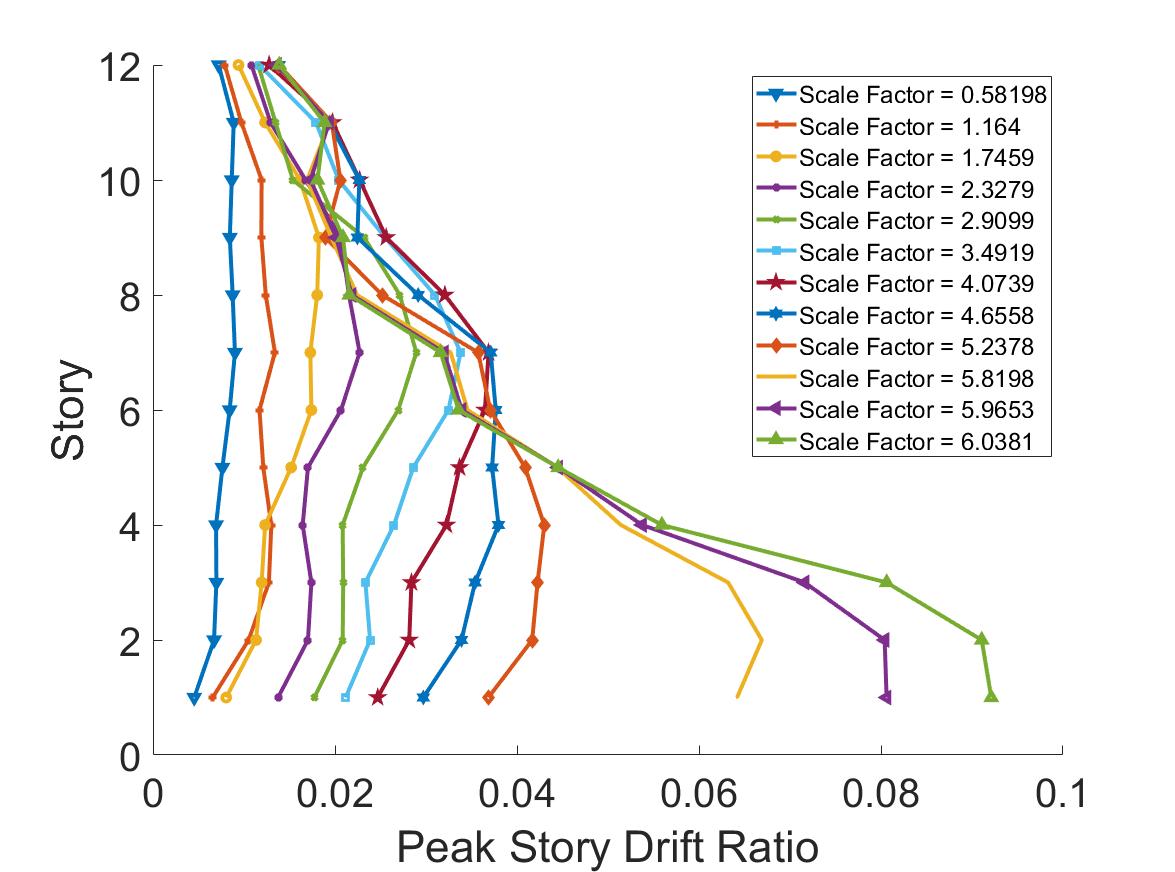}
		\caption{}
		\label{fig:psdr_sf}
	\end{subfigure}
	\caption{ The distributions of (a) peak absolute floor accelerations (PFAs) and (b) peak absolute story drift ratio (SDRs) changes with different ground motion intensities. The data is collected based on incremental dynamic analysis of a 12-story building based on the ground motion observed in Station Gilroy Array \#3 during 1989 Loma Prieta earthquake. Note that the figures only show the data collected from a single building subjected to the same earthquakes with different scaling factors. The real-world distribution changes would be much more complicated and intractable due to various earthquakes and building frames.}
	\label{fig:sf}
	\end{center}
\end{figure}

Here we give examples, using a 12-story building, to show how data distributions, including the distributions of peak absolute floor acceleration (PFAs) and peak absolute story drift ratio (SDRs) change with the ground motion intensities. Figure~\ref{fig:pfa_sf} shows the PFAs distribution changes with different scaled intensities of a seismic ground motion (the ground motion observed in Station Gilroy Array \#3 during 1989 Loma Prieta earthquake) in an incremental dynamic analysis. Different scale factors for incremental dynamic analysis indicate different ground motion amplitudes used for simulations. We use normal distribution to fit the density estimation of PFAs. With intensity increasing, both mean and variance of PFAs increases. Statistically, as the intensity increases, the support of PFA distribution spreads wider. When there are limited data which mostly fall into the overlapping supporting range, it would become more difficult to distinguish the distribution changes. Figure~\ref{fig:psdr_sf} shows how the peak story drift ratio changes with increasing ground motion intensities on each story of the 12-story building. With ground motion intensity increases, the SDRs of middle stories first increase. When intensities become higher, the lower stories ($1\sim5$ story) are severely damaged and finally collapse. This transition trend indicates that under different ground motion intensities, the changes of data distribution vary a lot at different stories.

In summary, Figure~\ref{fig:sf} shows that even with the same type of ground motions, the building exhibits distinct responses and damage patterns under different intensities. When different buildings are subjected to various types and intensities of earthquakes, the changes of their structural responses and damage patterns become more complicated and hard to model. Therefore we combine the ground motions of earthquakes as the input features for reducing the complexity of domain adaptation, which would be discussed in Section~\ref{sec:loss}.
%In conventional works, we often use single-floor-vibration-based features as input to the story-wise damage diagnosis models by assuming both training and test data are subjected to the same earthquake loading. 
%In the domain adaptation scenario, ignoring the earthquake loading information would mislead the model in finding earthquake-invariant features. Therefore, seismic ground motions are important information to  extract building-invariant features.
\subsubsection{Impact of Building Heights}

The structural damage pattern is closely related to various structural and non-structural components of buildings, including the type of the structural frames, stiffness, damping, and height. These physical properties affect the structural system's elastic and elasto-plastic behaviour under earthquakes and thus are key elements to determine the structural response-damage distributions~\cite{scawthorn1981seismic,cosenza2000damage,nam2006seismic}. These influences are often interdependent and complex. For example, input energy is not only related to the elastic period of the structure, but also the viscous damping and the characteristics of the plastic response~\cite{cosenza2000damage}. More importantly, we often lack the detailed prior knowledge about the physical properties of buildings. The lack of sufficient physical knowledge makes it impossible to explicitly model the impacts of physical properties on structural damage patterns.   

\begin{figure}[th!]
\begin{center}
	\begin{subfigure}{0.49\textwidth}
	\centering
		\includegraphics[scale=0.38]{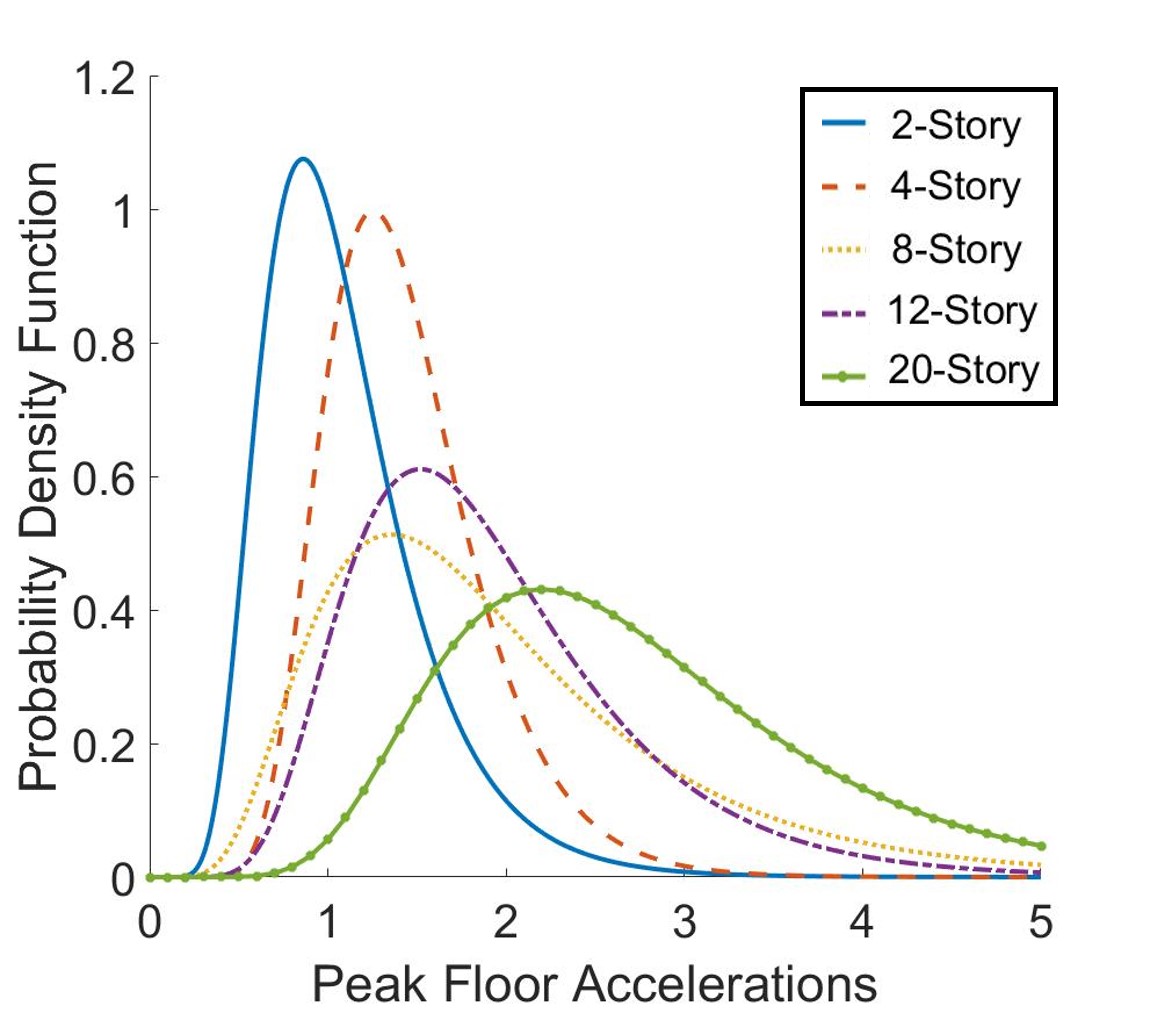}
		\caption{}
		\label{fig:pfa_story}
	\end{subfigure}
	\begin{subfigure}{0.49\textwidth}
	\centering
		\includegraphics[scale=0.38]{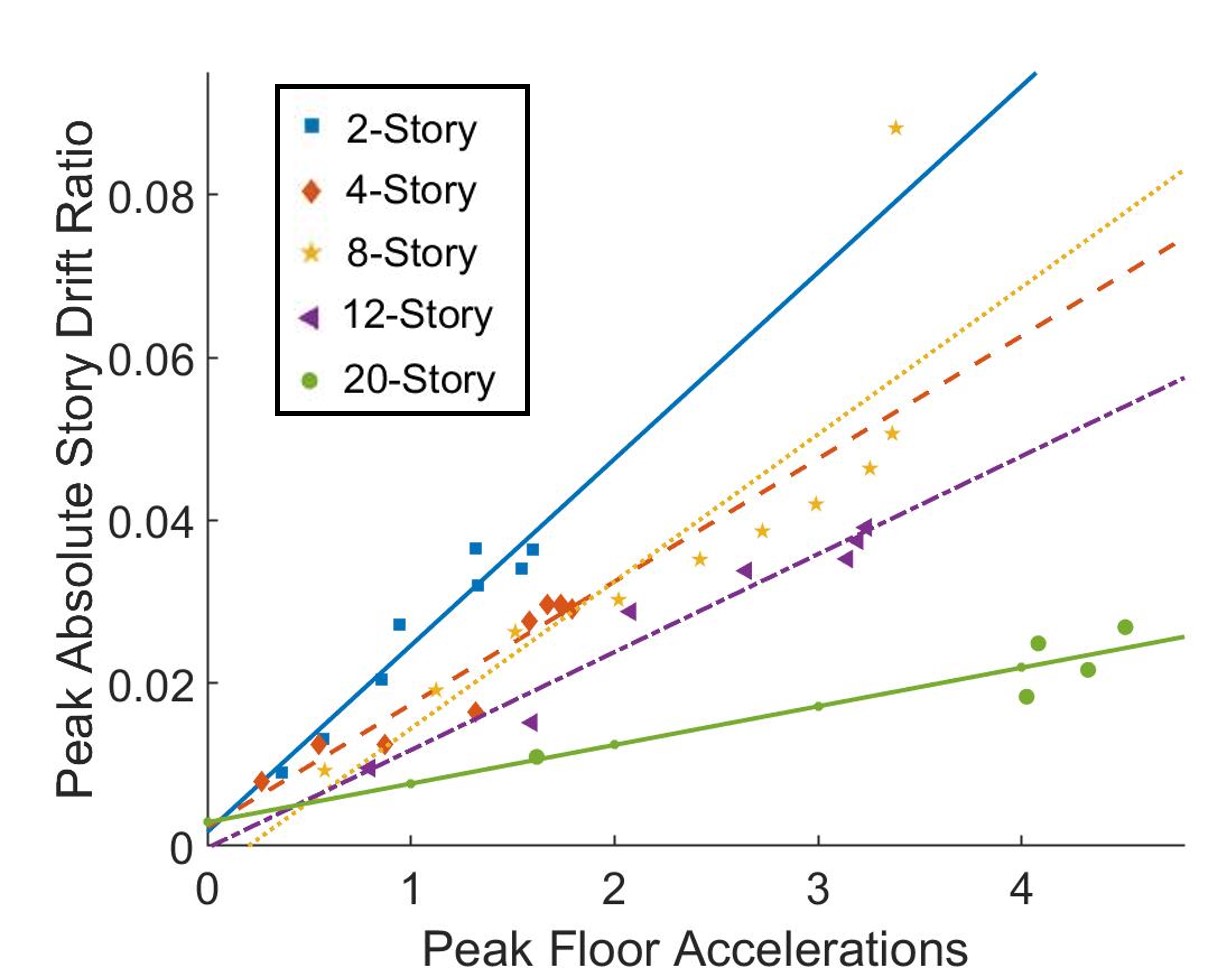}
		\caption{}
		\label{fig:psdr_story}
	\end{subfigure}
	\caption{The distributions of (a) peak absolute floor accelerations (PFAs) and (b) correlations between mean peak story drift ratio (SDRs) and peak absolute floor accelerations (PFAs) change with different building heights. The data is collected based on incremental dynamic analysis of 2-story, 4-story, 8-story, 12-story and 20-story buildings under different ground motion intensities.}
\label{fig:story}
\end{center}
\end{figure}

 Instead of constructing transfer functions using detailed prior physical knowledge, we utilize  simple characteristics which are easier to be obtained. In real-world practices, we may have some approximated and simplified physical knowledge, e.g., building heights, story/height ratio, and so on. It is difficult to directly learn the distribution transitions based on these simplified physical knowledge, but it is still helpful to quantitatively combine this physical knowledge to let it guide and improve the training of data-driven models. 
 
 As an example, we show how data distributions vary with building heights in Figure~\ref{fig:story}. We present the PFAs and mean SDRs of the $2$nd story of $2$-story, $4$-story, $8$-story, $12$-story and $20$-story buildings with Steel Moment-resisting Frames (SMFs) under different ground motions intensities. %with various scaled intensities observed in Canoga Park-Topanga Can station during 1994 Northridge earthquake
Figure~\ref{fig:pfa_story} shows the peak absolute floor acceleration distributions fitted with \emph{Log-Normal} distributions. It is shown that the distributions of structural responses are distinct from each other. The general trend is that at the same story, high-rise buildings tend to have more spread structural responses. Figure~\ref{fig:psdr_story} shows the correlations between PFAs and corresponding peak SDRs of different buildings. It is shown that low-rise buildings (e.g., 2-story and 4-story building) have larger SDRs, indicating more severe damages. This may be because the lower-rise buildings tend to experience a larger increase in ductility demands~\cite{anagnostopoulos1992investigation}. %Figure~\ref{fig:psdr_story} also shows that, under the same frame type, the buildings with similar heights tend to have similar data distributions.
A general trend is that the buildings with similar heights tend to have similar damage patterns given the same type of structural frame and the same ground motion. The similar physical properties indicate similar responses and deformation patterns given consistent earthquake loading. 
%If we had all prior physical knowledge about different buildings and ground motions, we can build transfer functions to learn the mapping from the structural responses to damage states. However in practices, the detailed physical knowledge are often difficult to obtain.
 In Section~\ref{sec:loss}, we show how we incorporate the building heights into the design of loss function for our data-driven model.

% \begin{figure*}[htp!]
% \begin{center}
% \includegraphics[scale = 0.2]{Figure/feature_dist_building.jpg}\label{fig:1a}
% \includegraphics[scale = 0.2]{Figure/feature_dist_eq.jpg}\label{fig:1b}
% \vspace{-0.2cm}
% \caption{The 2D tSNE visualization of data distribution of different buildings under earthquakes changes significantly. The red points represent data sampled from the source domain, the blue points represent data sampled from the target domain. The green line shows the classification decision boundary for source domain (red points). Cross indicates damaged and square indicates undamaged. (a) shows the difference of data distribution between 20-story building and 2-story building. (b) shows the difference of data distribution between 1994 Northridge earthquake and 1989 Loma Prieta earthquake. Both figures show that directly applying model trained on source domain dataset (red) to diagnose structural damage on target domain (blue) will result in bad performance.}
% \end{center}
% \label{fig:1}
% \end{figure*}

%data heterogeneity across different domains make it difficult to directly use raw building vibration signals as input. The heterogeneity includes: 1) the size of each domain's dataset varies significantly, from tens to thousands of samples; 2) the time length of each data sample varies significantly depending on the respective earthquake duration; 3) 

\section{Adversarial Knowledge Transfer Framework}\label{sec:framework}
In this section, we introduce a new intelligent building damage diagnosis framework which transfers the knowledge learned from multiple other buildings to predict and quantify the damage states of the target building without any label. As Figure~\ref{fig:framework} shows, our framework includes $3$ steps: data preprocessing, adversarial domain adaptation training, and damage diagnosis. Adversarial domain adaptation jointly trains $3$ neural networks, including domain-invariant feature extractor (E), damage predictor (M) and domain discriminator (D). The $3$ components are jointly optimized to obtain 1) domain-invariant and discriminative feature representations, and 2) a well-trained damage predictor which maps from the extracted features to damage states. Finally, in the stage of target building damage diagnosis, we directly input the extracted features for target domain data to well-trained damage predictor to diagnose the damage states of the target building. In this section, we first introduce our data preprocessing which prepares new preliminary inputs combining the earthquake information and structural responses in Section~\ref{sec:preprocess}. Then we present the details of adversarial domain adaptation training (Section~\ref{sec:adv_frame}), including problem background and key definitions (Section~\ref{sec:formulation}), architectures of each component (Section~\ref{sec:arch}), and our new physics-guided loss function (Section~\ref{sec:loss}). %Finally, we describe the adversarial training process for the proposed framework. With the adversarial training process, we would obtain the final feature extractor and damage predictor for final damage prediction and quantification. 
\begin{figure}[ht!]
    \centering 
    \includegraphics*[scale=0.135]{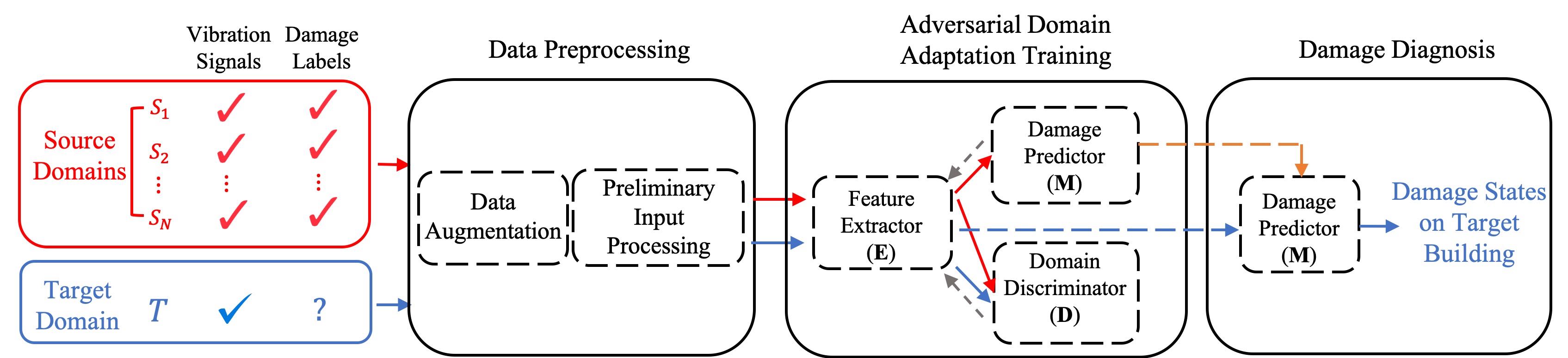}
    \caption{Our adversarial framework for multiple source domain adaptation for building damage diagnosis. Red represents source domains and blue refers to the target domain. \textcolor{black}{Grey dot line refers to back-propagation direction during training phase. The orange dot line means that we use the trained damage predictor as the final classifier for target damage prediction.} In our setup, there is no label available for target domain data, which is common in post-disaster scenarios due to the difficulty of label acquisition. \label{fig:framework}}
\end{figure}

%Our main objective is to learn and transfer the story-wise structural response-damages model extracted from multiple source buildings with labeled historical data to the target building without any labeled historical data. Here the input is story-wise structural vibration response, and label is the respective structural damage states. 

\subsection{Data Preprocessing}\label{sec:preprocess}
The focus of the data preprocessing includes two key steps: 1) preliminary input processing, and 2) data augmentation. 

\emph{Preliminary input processing}: Based on the characterization in Section~\ref{sec:challenge}, we know that the difference between data distributions of source domains and target domain is related to both earthquakes and building properties. 
%It is challenging to model the complicate mapping from these physical properties to distribution changes without any prior physical knowledge. Fortunately, in our problem, though buildings' physical properties are mostly unknown, the physical characteristics of earthquake excitations can be partially depicted by the collected ground motions. 
By including the ground motion information into the input, we can reduce the problem as learning building-invariant feature representations, by providing the information of earthquakes. Specifically, given floor vibration signals collected from source buildings and the target building, we extract preliminary inputs by combining the ground motion accelerations, the floor vibrations, and the ceiling vibration of each story. We first extract the frequency spectrum of each vibration signal. %Compared to directly using raw vibration data, using frequency domain information has less requirements on the consistent time length of each data sample (which is difficult to achieve in real-world practice), and is more flexible for a wider application. 
If the sampling rates are inconsistent across signals from different buildings, up-sampling and down-sampling are applied to unify the frequency resolution. Then we stack the frequency spectrums of the ground motion, the floor vibration and the ceiling vibration. Assuming the length of each frequency spectrum is $l$, we have each sample's input with size of $3\times l \times 1$. With stacked frequency information, the convolutional layers in later adversarial domain adaptation model can better understand the local correlation between ground motion and floor acceleration to extract domain-invariant features. 

\emph{Data Augmentation}: In real-world practice, the size of each domain's dataset varies significantly, from tens to thousands of samples. Training the deep neural network often needs a large amount of data to avoid over-fitting and biases induced by noises, as well as to ensure sufficient information for learning. Therefore, we conduct data augmentation on raw dataset to increase the data size. We use varying-length sliding window on the raw vibration signals. The vibration signals can be split into three stages: ambient excitation/free vibration before earthquake, earthquake-induced vibrations, ambient excitation/free vibration after earthquake. The rule of thumb to select the window length is to ensure that the duration of earthquake excitation is always covered. The window length varies to cover a part of vibration signals before earthquake happening and after earthquake happening. Then we smooth the signals and conduct the aforementioned preliminary input processing. The data augmentation increases the data amount and introduces different types of noises, which is important to improve the robustness of domain-invariant and discriminative feature representation learning.

\subsection{Adversarial Domain Adaptation}\label{sec:adv_frame}
This section introduces our adversarial learning algorithm for multiple source domain adaptation. We first provide a mathematical overview of domain adaptation problem. Then we present the detailed architectures of each component of the algorithm. Finally, we introduce the new physics-guided loss function for model optimization.

\subsubsection{Background and Definitions}\label{sec:formulation}

For simplicity, we denote the source domains as $S$ and target domain as $T$. Assume that we have $n \geq 2$ source buildings. Each of them has a distinct data distribution. We denote each source domain as $S_i\in S, i\in \{1,\cdots, n\}$. The training dataset in each source domain $S_i$ is composed of $S_i = \{ x_j^{S_i}, y_j^{S_i}\}_{j=1}^{m_{S_i}}$, where $x_j^{S_i}, y_j^{S_i}$ refers to the $j$th sample collected from the $i$th source building, and $m_{S_i}$ refers to the number of samples from the $i$th source building. Denote the sample space for $i$th domain as $\pazocal{X}_{S_i}$ where $x_j^{S_i}\in \pazocal{X}_{S_i}$, and damage state space for $i$th domain as $\pazocal{Y}_{S_i}$ where $y_j^{S_i}\in \pazocal{Y}_{S_i}$. Assume the labeled source data $S_i$ is drawn $i.i.d$ from the distribution $\pazocal{D}_{S_i}$. For the target building, we have a set of the collected earthquake-induced building floor vibration signal $\{x_j^T\}_{j=1}^{m_T}$, where $m_T$ refers to the number of instances we need to diagnose for the target domain. In this work, we assume that $m_T \ll \sum_i^n m_{S_i}$. Their corresponding true damage labels $\{y_j^T\}_{j=1}^{m_T}$ are unknown. Similar to the source domain definition, denote $x_j^{T}\in \pazocal{X}_{T}$ and  $y_j^{T}\in \pazocal{Y}_{T}$ for the target domain $T$. Denote the true data distribution on the target domain as $\pazocal{D}_{T}$. \textcolor{black}{As the damage label definition is consistent across all buildings~\cite{prestandard2000commentary}, it is reasonable to assume that source and target labels belong to the same label space $\pazocal{Y}$.}  Assume the unlabeled target sample $T$ is drawn $i.i.d$ from the marginal distribution of $\pazocal{D}_T$ over $X$, i.e. $\pazocal{D}_{T}^X$. A \emph{hypothesis} is a function/classifier $h: \pazocal{X}_T \rightarrow \pazocal{Y}_T$. The \emph{error}/\emph{risk} of a hypothesis on $\pazocal{D}$ is defined as:
\begin{align*}
    R_{\pazocal{D}} (h) = Pr_{(x,y)\sim \pazocal{D}} (h(x) \neq y).
\end{align*}

Our final goal is to predict the damage states of the target building without any labeled data collected from the target building. That is, without any information about the labels $\{y_j^T\}_{j=1}^{m_T}$, we aim to find an optimal hypothesis with a low target risk $R_{T} (h)$. 

In this problem, the data distribution varies between different source domains and between source domains and the target domain, i.e. $P(\pazocal{X}_{S_i},\pazocal{Y}_{S_i})\neq P(\pazocal{X}_{T},\pazocal{Y}_{T})\forall i$.
\textcolor{black}{\begin{definition}\textbf{$\pazocal{H}$-divergence}
(based on~\cite{kifer2004detecting}) Given a domain $\pazocal{X}$ with $\pazocal{D}$ and $\pazocal{D}'$ probability distributions over $\pazocal{X}$, let $\pazocal{H}$ be a hypothesis class on $\pazocal{X}$. Denote $I(h)$ the subsets of $\pazocal{X}$ which are the support of hypothesis $h$ in $\pazocal{H}$, i.e., $x\in I(h) \Leftrightarrow{} h(x)=1$. The divergence between $\pazocal{D}$ and $\pazocal{D}'$ based on $\pazocal{H}$ is
\begin{align*}
    d_{\pazocal{H}}(\pazocal{D}, \pazocal{D}') = 2\sup_{h\in \pazocal{H}}|Pr_{\pazocal{D}}[I(h)]-Pr_{\pazocal{D}'}[I(h)]|.
\end{align*}
\end{definition}
The ideal joint hypothesis class, $\pazocal{H}$, contains all the possible hypotheses over $\pazocal{X}$. In order to obtain the risk bound, $d_{\pazocal{H}}(\pazocal{D}, \pazocal{D}')$ can be further reduced to the total variation on $\pazocal{H}$. The variation can be depicted by the represent difference relative to other hypotheses in $\pazocal{H}$, which measures the adaptability of a source-trained classifier~\cite{ben2010theory}. \textit{Symmetric difference hypothesis space} $\pazocal{H}\Delta \pazocal{H}$ is defined for reasoning the represent difference in $\pazocal{H}$.
\begin{definition}
Symmetric difference hypothesis space $\pazocal{H}\Delta \pazocal{H}$ (\cite{ben2010theory}). For a given hypothesis class $\pazocal{H}$, the symmetric difference hypothesis space $\pazocal{H}\Delta \pazocal{H}$ is the set of hypotheses
\begin{align*}
    g \in \pazocal{H}\Delta \pazocal{H} \Longleftrightarrow g(x) = h(x) \oplus h'(x) \textnormal{  for }h, h'\in \pazocal{H}.
\end{align*}
\end{definition}
Where $\oplus$ refers to the XOR function. The hypothesis $g\in \pazocal{H} \Delta \pazocal{H}$ represents the set of disagreements between two hypotheses in $\pazocal{H}$.
% Let $\pazocal{H}$ denote the reproducing kernel Hilbert space (RKHS)~\cite{blitzer2006domain, ben2007analysis} for classifiers on $\pazocal{X}$. $\pazocal{H}\Delta \pazocal{H}$-divergence depicts the distance between two distributions on the symmetric difference hypothesis space $\pazocal{H}\Delta \pazocal{H}$~\cite{ben2010theory}.
%Define $\pazocal{H}$-divergence between the marginal distributions of the source and target domain on $X$ as $d_{\pazocal{H}}(\pazocal{D}_S^{X}, \pazocal{D}_T^{X})$. 
Define the \textit{empirical $\pazocal{H} \Delta \pazocal{H}$-divergence} as the $\pazocal{H} \Delta \pazocal{H}$-divergence between two sets of unlabeled samples $\hat{S}\sim (\pazocal{D}_S^X)^{m_S}$ and $\hat{T}\sim (\pazocal{D}_T^X)^{m_T}$, i.e. $d_{\pazocal{H}\Delta \pazocal{H}}(\hat{S}, \hat{T})$. $\pazocal{H} \Delta \pazocal{H}$-divergence refers to the represent differences on source and target data, which is closely related to the discrepancy between two data distributions~\cite{ben2007analysis}. }
Ben-David et.al proved that the final target classification risk is upper bounded by the combinations of \textit{empirical $\pazocal{H}\Delta \pazocal{H}$-divergence} $d_{\pazocal{H}\Delta \pazocal{H}}(\hat{S}, \hat{T})$ and the empirical risk on a single source domain $\hat{R}_S(h)$~\cite{ben2010theory}.

\begin{theorem}
\textbf{(Ben-David et al., 2006)} Let $\pazocal{H}$ be a hypothesis class of VC dimension $d$. Given the unlabeled samples $\hat{S}\sim (\pazocal{D}_S^X)^m$ and $\hat{T}\sim (\pazocal{D}^X_T)^{m}$, with probability $1-\delta$, for every function $h \sim \pazocal{H}$:
\begin{align*}
    R_{T}(h) \leq \hat{R}_S(h) +  d_{\pazocal{H}\Delta \pazocal{H}}(\hat{S}, \hat{T}) + \lambda + \pazocal{O}\left(\sqrt{ \dfrac{1}{m} (d\log \dfrac{m}{d}+\log \dfrac{1}{\delta})}\right)
\end{align*}
with $\lambda\geq \inf_{h^*}[R_{S}(h^{\star})+R_{T}(h^{\star})]$. 
\end{theorem}
In feature representation-based methods, the empirical risk on the source domain $R_{S}(h)$ depicts the discriminativeness of the damage classifier on the source domain $S$. $\hat{d}_{\pazocal{H}\Delta \pazocal{H}}(S,T)$ depicts the  domain-variance between the source data and the target data on $\pazocal{H}$. 

Extending this theorem to multiple source domains where $S_i\sim (\pazocal{D}_{S_i}^X)^m$ and $T\sim (\pazocal{D}^X_T)^{m}$, Zhao et.al~\cite{zhao2018adversarial} showed that $\forall w_i \in \mathbb{R}_{+}, \sum_i w_i = 1$, with probability $1-\delta$ and for every function $h \sim \pazocal{H}$,
\begin{align}
    R_{T}(h) \leq \sum_i w_i \left( R_{S_i}(h) +  d_{\pazocal{H}\Delta \pazocal{H}}(\hat{S}_i,\hat{T}) \right) + \lambda_w +\pazocal{O}\left(\sqrt{ \dfrac{1}{km} (d\log \dfrac{km}{d}+\log \dfrac{1}{\delta})}\right),\label{eq:1}
\end{align}
% Extending this theorem to multiple source domains where $\hat{S}_i\sim (\pazocal{D}_{S_i}^X)^m$ and $\hat{T}\sim (\pazocal{D}^X_T)^{m}$, we have Theorem~\ref{theorem:2}.
% \begin{theorem}Let $\pazocal{H}$ be a hypothesis class of VC dimension $d$. Given the samples $\hat{S}_i \sim (\pazocal{D}_{S_i}^X)^m$ and $\hat{T}\sim (\pazocal{D}^X_T)^{m}$,
% $\forall w_i \in \mathbb{R}_{+}, \sum_i w_i = 1$, with probability $1-\delta$ and for every function $h \sim \pazocal{H}$,
% \begin{align}
%     R_{T}(h) \leq \sum_i w_i \left( \hat{R}_{S_i}(h) +  d_{\pazocal{H}}(\hat{S}_i,\hat{T}) \right) + \lambda_{w} +\pazocal{O}\left(\sqrt{ \dfrac{1}{km} (d\log \dfrac{km}{d}+\log \dfrac{1}{\delta})}\right),\label{eq:1}
% \end{align}\label{theorem:2}
% \end{theorem}
where $d$ is the VC dimension of the hypothesis class and $k$ is a constant. To ensure a successful prediction of $h$ on the target domain, we need to minimize the target risk $R_{T}(h)$. To minimize the target risk, we need to minimize the empirical source risk $\hat{R}_{S_i}(h)$ and the empirical distribution distance between $S_i$ and $T$ on the input $\pazocal{X}$. That is, to ensure a successful knowledge transfer from multiple source domains to the target domain, we need to ensure the discriminativeness and the domain-invariance of the extracted features. Meanwhile, as there are multiple source buildings,  $w_i$ here represents a convex combination of the risk upper bounds based on different source building data distributions. 

\begin{figure}[ht!]
    \centering 
    \includegraphics*[scale=0.32]{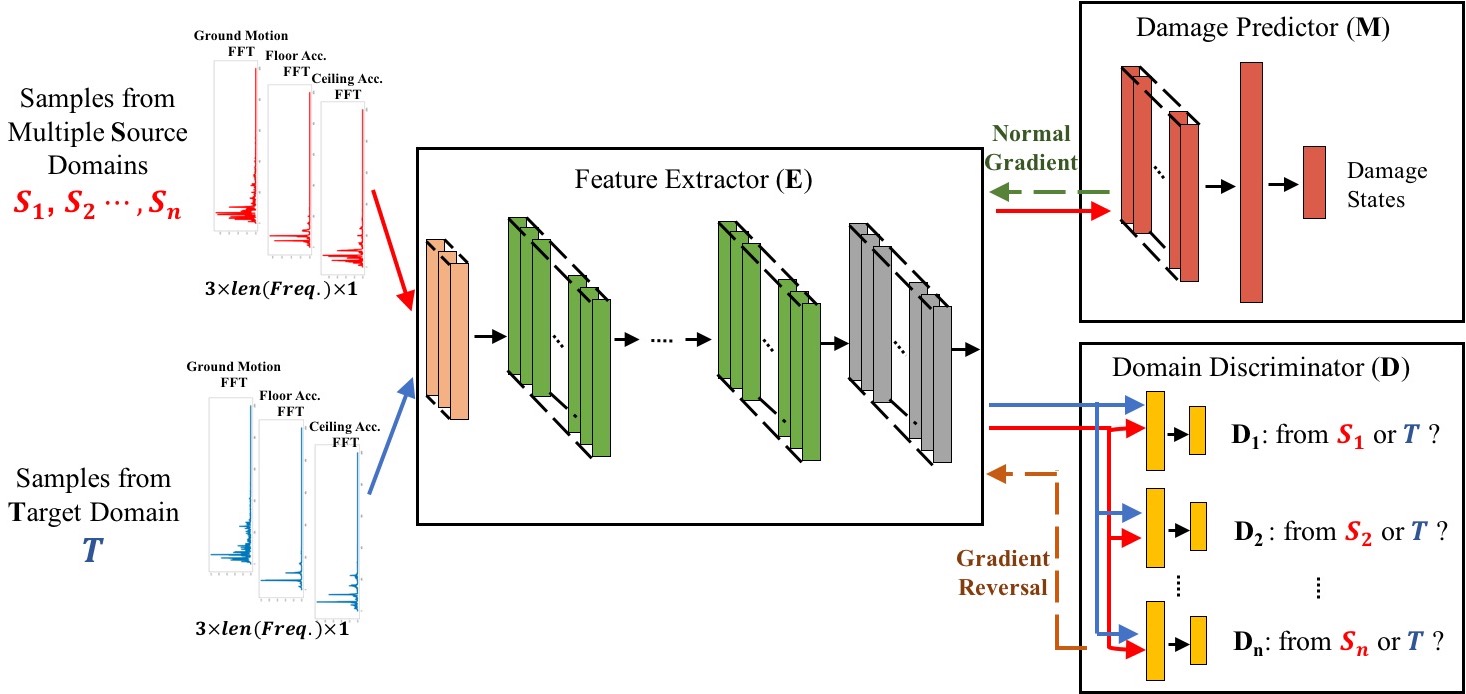}
    \caption{The overview of our adversarial domain adaptation algorithm for knowledge transfer from multiple source buildings to the target building.}
     \label{fig:framework_arch}
\end{figure}

\subsubsection{Architectures of Feature Extractor, Domain Discriminator, and Damage Predictor}\label{sec:arch}

Our deep adversarial domain adaptation framework includes three main components: feature extractor ($E$), damage predictor ($M$), and domain discriminator ($D$), as Figure~\ref{fig:framework_arch} shows. 
With the extracted frequency information of structural vibration responses (shown in Section~\ref{sec:preprocess}) as input, the feature extractor projects the input to a high-dimensional space to obtain the domain-invariant representations of structural responses. \textcolor{black}{We denote the extracted feature space as $\pazocal{Z}$. In the training stage, the extracted source feature $z_{S_i}$ is input to both domain discriminator and damage predictor, and  the target feature $z_{T}$ is input to domain discriminator, where $z_{S_i}, z_T \in \pazocal{Z}$.} Damage predictor is used to learn the optimal mapping from the extracted feature representations to damage states. Domain discriminator includes $n$ sub-classifiers. The $i$th sub-classifier focuses on distinguishing the difference between extracted source and target features. \textcolor{black}{For simplicity, we use a more general notation $\pazocal{F}$ to represent the mapping function of the three main components. We denote $\pazocal{F}_{E}(\cdot; \theta_e): \pazocal{X} \rightarrow \pazocal{Z}$ as the feature extractor neural network ($E$) with parameters $\theta_e$. Similarly, $\pazocal{F}_{M}(\cdot; \theta_m): \pazocal{Z}\rightarrow \pazocal{Y}$ refers to damage predictor $M$ with parameters $\theta_m$ which map extracted features to damage states (we assume there are $K$ damage states in total). $\pazocal{F}_{D_i}(\cdot; \theta_{d_i}): \pazocal{Z} \rightarrow \pazocal{C}$ refers to the $i$th domain discriminator $D_i$ with parameters $\theta_{d_i}$ mapping the extracted features to domain labels. $D_i$ takes extracted features from
the source domain $S_i$ and the target domain $T$. The domain label space $\pazocal{C}$ has two labels: $1$ represents that the sample comes from some source domain, $0$ represents that the sample is from the target domain.}

Our objectives are (1) to find an optimal feature extractor such the extracted features are domain-invariant as well as discriminative with respect to various damage states, and (2) to find an optimal damage predictor with the extracted features as input. To achieve  objective (1), our algorithm optimizes feature extractor and domain discriminator in an adversarial way, such that even the best-trained domain discriminator cannot tell the difference between the extracted features across source and target domains. To achieve objective (2), our algorithm jointly optimizes the damage predictor together with the feature extractor to build accurate the structural damage prediction model given extracted feature representations. In this way, we can find an optimal trade-off between the domain-invariance and discriminativeness of extracted features, and thus successfully adapt the knowledge learned from other buildings to help diagnose the target building of interest. \textcolor{black}{When there is partially labeled target data available, for example, undamaged labels collected from different earthquake excitation, the framework could be directly adopted by jointly training the damage predictor using both labeled source data and partially labeled target data.}

The architectures of feature extractor, domain discriminator and damage predictor are all based on deep convolutional neural network (CNN). In recent years, deep learning techniques benefit many applications in engineering fields. By constructing neural networks with deep and special architectures, we can approximate a wide range of highly non-linear and complex mapping functions~\cite{lecun2015deep}. A variety of deep neural networks have been proved to be effective and powerful in many real-world tasks, including structural health monitoring~\cite{jang2019deep,ye2019review,bao2019computer}, computer vision~\cite{krizhevsky2012imagenet}, natural language processing~\cite{devlin2018bert}, medical imaging~\cite{smailagic2018medal, smailagic2019medal} and video game~\cite{mnih2013playing}. Some typical architectures of deep neural networks include deep convolutional neural network (CNN)~\cite{krizhevsky2012imagenet} and recurrent neural network~\cite{mikolov2010recurrent}. Among these architectures, deep convolutional neural network combines convolutional layers to understand the local structures of features in various resolutions and thus becomes very powerful to learn effective representations from complex data.

\begin{table}[ht]
\caption{Architecture for 5-class damage quantification}
\begin{center}
\begin{tabular}{rccccc}
\hline
Networks&Operation & Kernel & Strides & Feature maps & Activation\\
\hline
\hline
\multirow{5}{*}{Feature Extractor}&
Input&  & &$3\times 1000 \times 1$&\\
&Convolution & $81\times 5\times 1$ & $2\times 1$ &$81\times 499 \times 1$&LeakyReLU\\
&Convolution & $81\times 5 \times 1$ & $2\times 1$ &$81\times 248 \times 1$&LeakyReLU\\
&Convolution & $81\times 3 \times 1$ & $2\times 1$ &$81\times 124 \times 1$&LeakyReLU\\
&Convolution & $81\times 3 \times 1$ & $2\times 1$ &$81\times 61 \times 1$& \\
\hline
\multirow{5}{*}{Damage Predictor}
&Convolution & $243\times 3\times 1$ & $2\times 1$ &$243\times 30 \times 1$&LeakyReLU\\
&Convolution & $81\times 3\times 1$ & $1\times 1$ &$81\times 29 \times 1$&LeakyReLU\\
&Convolution & $27\times 3\times 1$ & $1\times 1$ &$27\times 28 \times 1$&LeakyReLU\\
&Flatten&&&756&\\
&Full connection&&&5&Softmax\\
\hline
\multirow{2}{*}{\shortstack[r]{The $i$th \\ Domain Discriminator}}
&Flatten&&&4941&\\
&Full connection&&&2&Softmax\\
\hline
\end{tabular}
\end{center}
\label{table: Architecture}
\end{table}

\emph{Feature Extractor:}
Feature extractor focuses on extracting domain-invariant and discriminative features for all source and target domains. The feature extractor take the stacked frequency information from the structural responses and respective ground motion as input.  With three channels including ground motion, floor and ceiling vibration in the input, we use multiple convolutional layers with varying kernel sizes to enlarge the number of channels to discover the detailed local correlation between the frequencies of ground motion and floor vibrations, as well as to extract the combinations of frequency energy with varying resolutions. The selection of kernel size obeys the rule of thumb that larger size at the bottom layers and smaller size at the top layers, which can better capture the information at different resolution scales~\cite{lecun1998gradient,zeiler2014visualizing,luo2016understanding}. The output of feature extractor is an input to both the damage predictor and the domain discriminator. 

The optimization of the feature extractor takes the gradient information from both domain discriminator and damage predictor. This optimization is a trade-off between domain-invariance and damage discriminativeness. The extracted features are designed to improve the performance of the damage predictor and to be highly domain-invariant such that even optimal domain discriminator cannot distinguish which domain they come from. To achieve this, we use the newly designed physics-guided loss function to train the feature extractor via gradient reversal, which is introduced in Section~\ref{sec:loss}.

\emph{Damage Predictor:}
The damage predictor is designed to ensure the discriminativeness of the extracted features. It takes the extracted features as inputs and classifies the samples into different damage states. The basic intuition for the damage predictor is to model the mapping from extracted features to damage states. A well-trained damage predictor back-propagates as much sufficient information as possible to the feature extractor, to improve the discriminativeness of extracted features. When designing the damage predictor, the expressive power of the model needs to be ensured to model highly non-linear functional relationship. However, the damage predictor cannot be too deep, otherwise the vanishing gradient back-propagated to the feature extractor no longer provides any effective information about the damage patterns. According to our empirical study, we design the damage predictor as $3$ convolutional layers with shrinking kernel size followed by $1$ fully connected layer. 

\emph{Domain Discriminator:}
Domain discriminator examines the domain-invariance of the extracted features.
Each domain discriminator $D_i$ takes only the extracted features from the source domain $S_i$ or the target domain $T$ as inputs, and classifies the samples into $2$ classes: the sample comes from the source domain $S_i$ or the sample comes from the target domain $T$. We have multiple sub-classifiers inside the domain discriminator to ensure the consistent distribution between extracted features of each source domain and the target domain. 
The domain discriminator is trained adversarially with the feature extractor, such that the domain discriminator which can best distinguish the domain-difference between features is still confused by the domains of the feature extractor's outputs. The design of domain discriminator should be simpler than that of the damage predictor -- otherwise, the domain discriminator will dominate the training process and prevent the damage discriminativeness of the extracted features.

Table~\ref{table: Architecture} provides the details of architecture we used in our experiment. The architecture may change with different knowledge transfer tasks. For example, for a binary damage detection task, the depth of feature extractor and damage predictor could be reduced since the damage detection task is relatively simpler than the damage quantification task. 

\subsubsection{Physics-guided Loss Function for Adversarial Domain Adaptation} \label{sec:loss} 

To best learn domain-invariant feature representations, we train the feature extractor and domain discriminator in an adversarial way. The key idea of adversarial training is to play a minimax two-player game between the feature extractor and domain discriminator. In detail, we first optimize the domain discriminator to best distinguish the domain-difference of any features output by the feature extractor, and then we optimize the feature extractor such that the domain discriminator is not able to differentiate features as from the source or target domain. By alternatively optimizing the feature extractor and domain discriminator, the adversarial training aims at finding a saddle point such that the two networks achieve an equilibrium. In this way, we obtain an optimal feature extractor to best learn the domain-invariant feature representations. In recent years, adversarial learning has attracted much attention in machine learning communities since it avoids complex sampling and inference and enables an efficient learning of generative models~\cite{goodfellow2014generative, tzeng2017adversarial}. Studies based on adversarial learning, such as Generative Adversarial Nets (GANs) and adversarial defenses, have led to significant advances in generative learning and robustness of deep neural networks. 

In this section, we first present the adversarial loss we designed for model optimization. We then provide theoretical insights behind the physics-guided weight design. As an important component of the loss function design, our physics-guide weight is designed to provide a better guarantee for damage prediction error.

\vspace{0.2cm}
{ \noindent \emph{I. Physics-guided Adversarial Loss} \par} 
\vspace{0.1cm}

A key challenge is the loss function design. The optimization objective is to simultaneously (1) minimize the loss of damage prediction by optimizing both feature extractor and damage predictor, (2) minimize the loss of domain discrimination by optimizing domain discriminator, and (3) maximize the loss of domain discrimination by optimizing the feature extractor. Different from the two-player game in conventional adversarial learning, we also have damage predictor to ensure the damage discriminativeness of the extracted features, which makes the optimization problem a three-player game. Besides, each module is composed of a deep neural network. Jointly optimizing the three non-convex complex networks with unknown landscape is a challenging task. Therefore, the loss function design is important to ensure the performance of the adversarial knowledge transfer framework. 

 We use $\pazocal{L}_M^i$ to represent the cross-entropy loss of using the extracted features to predict the damage states for source domain $S_i$. $\pazocal{L}_{D_i}$ represents the cross-entropy loss of distinguishing the extracted features from $S_i$ or $T$. $\lambda$ is defined as the factor to tune the trade-off between domain-invariance and discriminativeness of the extracted features. 

The feature extractor, damage predictor, and domain discriminator are optimized by solving the following optimization problem,
\begin{align}
   \min_{\theta_e, \theta_m}~~~\left[\sum_{i=1}^n w_i \pazocal{L}^i_M(\theta_e, \theta_m)- \min_{\theta_{d_1} ,\cdots, \theta_{d_n}}\lambda \sum_{i=1}^n w_i \pazocal{L}_{D_i}(\theta_e, \theta_{d_i})\right],\label{eq:loss}
\end{align}
where
\begin{align}
    \pazocal{L}^i_M(\theta_e, \theta_m)&=- \mathbb{E}_{(\mathbf{x}, y)\sim \pazocal{D}_{S_i}}\sum_{k=1}^K I(y=k)\log\pazocal{F}_M(\pazocal{F}_E(\mathbf{x})) \\
    \pazocal{L}_{D_i}(\theta_e, \theta_{d_i})& = -\mathbb{E}_{\mathbf{x}_s\sim \pazocal{X}_{S_i}}\left[\log \left(\pazocal{F}_{D_i}(\pazocal{F}_E(\mathbf{x}_s))\right) \right] -\mathbb{E}_{\mathbf{x}_t\sim \pazocal{X}_{t}}\left[\log\left(1- \pazocal{F}_{D_i}(\pazocal{F}_E(\mathbf{x}_t))\right) \right].
\end{align}
$w_i$ is the weight for source domain $i$, which is designed to balance the influence of different source buildings, as mentioned in Section~\ref{sec:formulation}. We introduce a new physics-guided weight design which measures the importance of each source building using their physical similarities to the target building. We present the details of this physics-guided weight in the next subsection.

The problem in~(\ref{eq:loss}) is a minimax problem. Since the goals of the feature extractor and of the domain discriminator are conflicting, we need to train them in an adversarial manner to find the saddle point $\hat{\theta}_e, \hat{\theta}_m, \hat{\theta}_{d_i}$ such that
\begin{align}
    (\hat{\theta}_e, \hat{\theta}_m) &= \arg\min \left[\sum_{i=1}^n w_i \pazocal{L}^i_M(\theta_e, \theta_m)- \min_{\theta_{d_1} ,\cdots, \theta_{d_n}}\lambda \sum_{i=1}^n w_i \pazocal{L}_{D_i}(\theta_e, \theta_{d_i})\right]\\
    (\hat{\theta}_{d_1}, \cdots, \hat{\theta}_{d_n}) &= \arg\min \lambda \sum_{i=1}^n w_i \pazocal{L}_{D_i}(\theta_e, \theta_{d_i}).
\end{align}
In practices, to find a stationary saddle point, the adversarial training is achieved by a non-trivial gradient-reversal layer. As Figure~\ref{fig:framework_arch} shows, the gradient-reversal layer connects the feature extractor and the domain discriminator. In back-propagation, it reverses the gradient by multiplying it by a negative scalar during back-propagation~\cite{ganin2016domain}. That is, during back-propagation, the neural networks follow the gradient updates as:
\begin{align}
    \theta_e &\leftarrow \theta_e - \delta \sum_i\left(\dfrac{\partial w_i \pazocal{L}_M^i}{\partial \theta_e} - \lambda \dfrac{\partial w_i \pazocal{L}_{D_i}}{\partial \theta_e}\right)\\
    \theta_m &\leftarrow \theta_m - \delta\sum_i\left(\dfrac{\partial w_i \pazocal{L}_M^i}{\partial \theta_m}\right)\\
    \theta_{d_i} &\leftarrow \theta_{d_i} - \delta\left(\lambda \dfrac{\partial w_i \pazocal{L}_{D_i}}{\partial \theta_{d_i}}\right)
\end{align}

\vspace{0.2cm}
{\noindent \emph{II. Theretical Insights behind Physics-guided Weights} \par} 
\vspace{0.1cm}
An important component of the loss function is the weight $w_i$. This weight decides how efficiently we can integrate the knowledge from different source domains. From the Equation~\ref{eq:1}, an optimal $w_i$ for different source domains can approximate the tightest upper bound of the target risk.  In previous work, people mostly set $w_i = 1/n$ where $n$ is the number of the source domains~\cite{zhao2018adversarial}. However, we show that $w_i =\dfrac{1/d_{\pazocal{H}\Delta \pazocal{H}}(\hat{S}_i, \hat{T})}{\sum_i 1/d_{\pazocal{H}\Delta \pazocal{H}}(\hat{S}_i, \hat{T})}$ results in a tighter bound, which provides an improved guarantee on the performance of knowledge transfer.

%In the adversarial framework, the main goal is to regularize the learning of domain-invariant mapping to minimize the distribution distances between the source and target domains. The damage predictor aims at minimizing damage prediction error on the source domains, i.e. $R_{S_i}(h)$, where $h= \pazocal{F}_M(\pazocal{F}_E(\cdot;\theta_e);\theta_m)$. The domain discriminator aims at encouraging the minimization of distribution differences between the source and target domains, i.e. $d_{\pazocal{H}\Delta \pazocal{H}}(\hat{S}_i, \hat{T})$. 
%The objective of training feature extractor is two-fold: 1) to ensure the discriminativeness, the extracted features should minimize the loss of damage predictor; 2) to ensure the domain-invaraince, the extracted features should maximize the loss of domain discriminator. 

\begin{theorem}
Given $\hat{R}_{S_1}(h) = \cdots =  \hat{R}_{S_n}(h) \forall i$,
$w_i =\dfrac{1/d_{\pazocal{H}\Delta \pazocal{H}}(\hat{S}_i, \hat{T})}{\sum_i 1/d_{\pazocal{H}\Delta \pazocal{H}}(\hat{S}_i, \hat{T})}$ results in a tighter upper bound for the target risk  $R_{T}(h)$ than $w_i = 1/n$.
\end{theorem}
\begin{proof}
Given that $R_{S_i}(h) $ depends on the expressive power of the shared damage classifier $h= \pazocal{F}_M(\pazocal{F}_E(\cdot;\theta_e);\theta_m)$, we assume that the classifier is optimally trained and over-fitted on the source data and the source risks are equal to one another, i.e., $\hat{R}_{S_1}(h) = \cdots =  \hat{R}_{S_n}(h) \forall i$. Therefore, for $d_{\pazocal{H}\Delta \pazocal{H}}(\hat{S}_i, \hat{T}) > 0$,
\begin{align*}
    \sum_i &\dfrac{1/d_{\pazocal{H}\Delta \pazocal{H}}(\hat{S}_i, \hat{T})}{\sum_i 1/d_{\pazocal{H}\Delta \pazocal{H}}(\hat{S}_i, \hat{T})}\left( \hat{R}_{S_i}(h) +  d_{\pazocal{H}\Delta \pazocal{H}}(\hat{S}_i, \hat{T}) \right) \\
    & = \sum_i \dfrac{1/d_{\pazocal{H}\Delta \pazocal{H}}(\hat{S}_i, \hat{T})}{\sum_i 1/d_{\pazocal{H}\Delta \pazocal{H}}(\hat{S}_i, \hat{T})} \hat{R}_{S_i}(h) + \dfrac{n}{\sum_i 1/d_{\pazocal{H}\Delta \pazocal{H}}(\hat{S}_i, \hat{T})}\\
    &\leq  \hat{R}_{S_i}(h) + \sum_i\dfrac{1 }{n}d_{\pazocal{H}\Delta \pazocal{H}}(\hat{S}_i, \hat{T})\forall i \text{ (based on Cauchy-Schwarz inequality)}\\
    &= \sum_i \dfrac{1}{n} \left(\hat{R}_{S_i}(h) + d_{\pazocal{H}\Delta \pazocal{H}}(\hat{S}_i, \hat{T})\right).
\end{align*}
It has been shown that $\sum_i w_i \left(R_{S_i}(h) + d_{\pazocal{H}\Delta \pazocal{H}}(\hat{S}_i, \hat{T})\right)$ decides the upper bound of the target risk, i.e., $R_{T}(h)$ in Equation~\ref{eq:1}. Therefore, $w_i =\dfrac{1/d_{\pazocal{H}\Delta \pazocal{H}}(\hat{S}_i, \hat{T})}{\sum_i 1/d_{\pazocal{H}\Delta \pazocal{H}}(\hat{S}_i, \hat{T})}$ renders a tighter upper bound of the target risk.
\end{proof}
Therefore, when designing the loss function for the framework, it is necessary to weigh the loss of different domains unequally to achieve better knowledge transfer performance. However, as we have no labeled data for the target building. it is unlikely to directly estimate the target data distribution and further measure the empirical distribution divergence $ d_{\pazocal{H}\Delta \pazocal{H}}(\hat{S}_i, \hat{T})$. 

One solution to approximate the data distribution divergence is to embed physical understanding of earthquake-induced structural damage patterns variations. In Sections~\ref{sec:challenge} and \ref{sec:formulation}, we show that the divergence between two buildings' sample data distributions is related to the similarity between the two building's physical properties. thus, we use $d_{\pazocal{H}_u\Delta \pazocal{H}_u}(\pazocal{U}_{S_i},\pazocal{U}_T)$ to indicate the relative changes of $d_{\pazocal{H}\Delta \pazocal{H}}(\hat{S}_i, \hat{T})$. The optimal weight for the $i$th domain, $w_i$, is designed as follows,
\begin{align}
    w_i &\propto \dfrac{f(1/d_{\pazocal{H}_u\Delta \pazocal{H}_u}(\pazocal{U}_{S_i},\pazocal{U}_T))}{\sum_i f(1/d_{\pazocal{H}_u\Delta \pazocal{H}_u}(\pazocal{U}_{S_i},\pazocal{U}_T))},
\end{align}
where $f$ is a non-decreasing function to model the relationship between $d_{\pazocal{H}_u\Delta \pazocal{H}_u}(\pazocal{U}_{S_i},\pazocal{U}_T)$ and $d_{\pazocal{H}\Delta \pazocal{H}}(\hat{S}_i, \hat{T})$.

Intuitively, the knowledge learned from the similar buildings should be (1) more efficiently transferred to the target building, and (2) more informative and indicative of the damage prediction on the target building. Therefore, we assign higher weights to the source domains which have more similar physical properties in the loss design. 

In practices, most available physical knowledge about buildings are simple and fuzzy physical properties. Denote $U_{S_i}$ and $U_T$ as the known physical knowledge about the $i$th source building and the target building $T$. From empirical study, we found that taking $f$ as exponential functional family performs best. Therefore, we calculate $w_i$ by taking the reciprocal of the similarity between  $U_{S_i}$ and $U_T$, and then normalize $w_i$ across all source domains using softmax function as 
\begin{align}
    w_i &= \dfrac{\exp{\left[1/dist(U_{S_i}, U_{T})\right]}}{\sum_i \exp{\left[1/dist(U_{S_i}, U_{T})\right]}}.\label{eq:w1}
\end{align}
\textcolor{black}{Softmax normalization is a classical way to normalize the weights in reweighting schemes, for example, attention mechanisms~\cite{vaswani2017attention}. Compared to linear transformation, it is known as a smoother approximation to the maximum function, giving the most important source domains larger weights \cite{goodfellow2016deep}. Meanwhile, with exponential functions, the outputs of softmax normalization are more stable to variation of inputs~\cite{goodfellow2016deep}.Therefore, we use softmax to conduct the weight normalization.} By using softmax normalization, we can smooth the influence of physical properties' variations. The physical knowledge is combined to regularize the domain-invariant and discriminative feature extraction. \textcolor{black}{ The similarities between different buildings' physical properties is approximated by the similarities between known physical parameters can be obtained (e.g., building heights, number of building stories, span of bridge, first-mode periods of bridge, designed strength of bridge, and other quantified physical properties of mechanical systems)}. The key is the insight that using physical similarities guides the knowledge transfer between buildings.  We define the similarity as a function of structural properties:
\begin{align}
    dist(U_{S_i}, U_{T}) = \left(1-U_{S_i}/U_T\right)^2 +\epsilon,\label{eq:w2}
\end{align}
where $\epsilon$ is a smoothing factor to avoid $dist = 0$. \textcolor{black}{The first principle to select physical parameters is that there should be existing evidence showing these properties would significantly impact the data distributions, for example, structural damage patterns in our earthquake-induced building damage diagnosis scenario. This knowledge could be obtained through empirical studies and previous experiments. Besides, the selection of weights is also related to whether the weight distribution weighs different source domains in a reasonable way and whether the validation results are good.  As mentioned in Section~\ref{sec:challenge}, how building physical property influences data distribution is a complex process. In practices, limited physical knowledge may constrain the accuracy to approximate real distribution divergence. However, this rough estimation is still helpful to leverage different source domains and guide the training process. Besides, data-driven model also automatically updates their parameters to complement insufficient divergence approximation by physics-guided weights. }

%In the proposed framework, we train domain discriminator $D$ to best distinguish extracted features from multiple source domains and the target domain. Meanwhile, we train $E$ to maximize the loss of domain discriminator such that the best domain discriminator cannot tell the difference between features from source domains and the target domain. In other world, $E$ and $D$ are playing a two-player minimax game~\cite{goodfellow2014generative} with the loss function. Meanwhile, we train $E$ and $M$ to ensure that the loss of label prediction is minimized such that the extracted feature representations are still effective for structural damage estimation. Therefore, by jointly optimizing $E,M,$ and $D$, we can obtain a best projection for domain-invariant and damage-discriminative feature extraction.

% \begin{figure}[t!]
%     \centering 
%     \includegraphics*[scale=0.43]{Figure/kernel.jpg}
%     \caption{The proposed physical kernel-based loss function.}
% \end{figure}

\section{Evaluation}\label{sec:eval}
In this section, we evaluate our framework on both simulation and real-world experimental earthquake-induced building vibration datasets. \textcolor{black}{The main objective is to integrate and transfer the knowledge learnt from multiple source buildings under different earthquakes to help diagnose a new building's earthquake-induced structural damages. The evaluations are twofold. We first evaluate the framework on transferring knowledge across five different buildings through numerical simulations. In each experiment, we select one building as the target building and the remaining as source buildings. The experiments are repeated for five times to evaluate the performance of our framework on each different target building. Next, we further evaluate the framework on transferring from numerical simulation data to real-world experimental data. In this experiment, five simulated buildings are source buildings, and a real-world experimental structure subjected to a series of earthquake ground motions is the target building.}  
In this section, we first give a brief description about the datasets in Section~\ref{sec:data}. Then we describe the baseline methods for comparison (Section~\ref{sec:base}), the knowledge transfer performance on simulation data (Section~\ref{sec:simulate}) and experimental dataset (Section~\ref{sec:realdata}). Finally, we characterize the training process and discuss the effect of hyperparameter $\lambda$ (Section~\ref{sec:hyperparam}).

\subsection{Data Description}\label{sec:data}
To evaluate the performance of our algorithm, we first transfer the knowledge between simulation data, and then transfer the knowledge learned from simulation data to diagnose the structural damages in shake-table experiments. 

\textit{Simulation Data:} To understand the building structural damage patterns under earthquakes, our building response database consists of a wide range of archetype steel frame buildings with MRFs~\cite{elkady2014modeling,elkady2015effect,hwang2017earthquake,hwang2018nonmodel}. The archetypes we used include 2-story, 4-story, 8-story, 12-story and 20-story building with a first-story height of 4.6m and a typical story height of 4m. The steel MRFs of these archetypes are designed with three strong-column/weak-beam (SCWB) ratios of 1.0~\cite{medina2004seismic}. The researchers utilize a suite of ground motions with large moment-magnitude ($6.5\leq M_w\leq 7$) and short closest-to-fault-rupture distance ($13km <R_{rup}<40km$). These ground motions are collected from 40 observation stations during 5 previous earthquake events. To simulate the building responses, two-dimensional nonlinear model representations of all the archetype MRFs are developed in the Open System for Earthquake Engineering Simulation (OpenSEES) Platform~\cite{mckenna1997object,hwang2018nonmodel}, and incremental dynamic analysis (IDA)~\cite{vamvatsikos2002incremental} is performed. The floor accelerations and story drift ratios are recorded with sampling rates of $50Hz \sim 200Hz$ for different ground motions over a wide range of incremental factors. \textcolor{black}{ The sampling rates are not consistent because (1) real-world earthquake ground motions are utilized and their sampling rates are different and (2) in the OpenSEES platform, a special collapse subroutine is used to avoid the non-convergence problem, which causing irregular sampling rates. More detailed about the simulation data could be found in~\cite{elkady2014modeling, elkady2015effect, hwang2017earthquake, hwang2018nonmodel}. }

\textit{Experimental data: } A series of shake table tests of a 1:8 scale model for a 4-story steel MRFs are conducted at the State University of New York at Buffalo~\cite{lignos2008shaking}. The structure is subjected to a series of the scaled 1994 Northridge earthquake ground motions recorded at the Canoga Park, CA, station. The scale factor ranges from 0.4 to 1.9. Accelerometers and displacement meters are instrumented on the structure to record the structural responses and story drift ratios 128Hz. 

\textit{True damage label determination: } According to the current standard (FEMA P695)~\cite{prestandard2000commentary,applied2005improvement,ramirez2009building, kircher2010evaluation}, structural damage states are defined based on the ground-truth peak story drift ratio at each story. For the damage detection task, we divide the damage state into no damage ($SDR \in [0, 0.01)$) and damaged($SDR \in [0.01,+\infty)$). For the task of damage quantification, we use $5$ damage states: no damage ($SDR \in [0, 0.01)$),  slight damage ($SDR \in [0.01, 0.02)$), moderate damage ($SDR \in [0.02, 0.03)$), severe damage ($SDR \in [0.03, 0.06)$),  and collapse ($SDR \in [0.06,+\infty)$).

% \begin{figure}[ht!]
% \begin{center}
% \begin{subfigure}{\textwidth}
% 	\centering 
% 		\includegraphics[scale=0.15]{Figure/12_story_results.jpg}
% 		\caption{}
% 		\label{fig:2a}
		
% 	\end{subfigure}
% 	\begin{subfigure}{0.49\textwidth}
% 	\centering 
% 		\includegraphics[scale=0.095]{Figure/2_story_results.jpg}
% 		\caption{}
% 		\label{fig:2b}
% 	\end{subfigure}
% 	\begin{subfigure}{0.496\textwidth}
% 	\centering 
% 		\includegraphics[scale=0.0965]{Figure/4_story_results.jpg}
% 		\caption{}
% 	\label{fig:2c}
% 	\end{subfigure}\\
% 	\begin{subfigure}{0.492\textwidth}
% 	\centering 
% 		\includegraphics[scale=0.096]{Figure/8_story_results.jpg}
% 		\caption{}
% 		\label{fig:2d}
% 	\end{subfigure}
% 	\begin{subfigure}{0.496\textwidth}
% 	\centering 
% 		\includegraphics[scale=0.0965]{Figure/20_story_results.jpg}
% 		\caption{}
% 		\label{fig:2e}
% 	\end{subfigure}
% 	\caption{(a), (b), (c), (d), (e) compare the domain adaptation performance between our method and other approaches on binary damage detection (blue) and 5-class damage quantification tasks (yellow). We use the dotted line to represent the damage prediction accuracy of directly training on the target domain as reference. }
% 	\label{fig:result2}
% 	\vspace{-1cm}
% \end{center}
% \end{figure}
\begin{figure}[htp!]
\begin{center}
\begin{subfigure}{0.505\textwidth}
	\centering 
		\includegraphics[scale=0.0445]{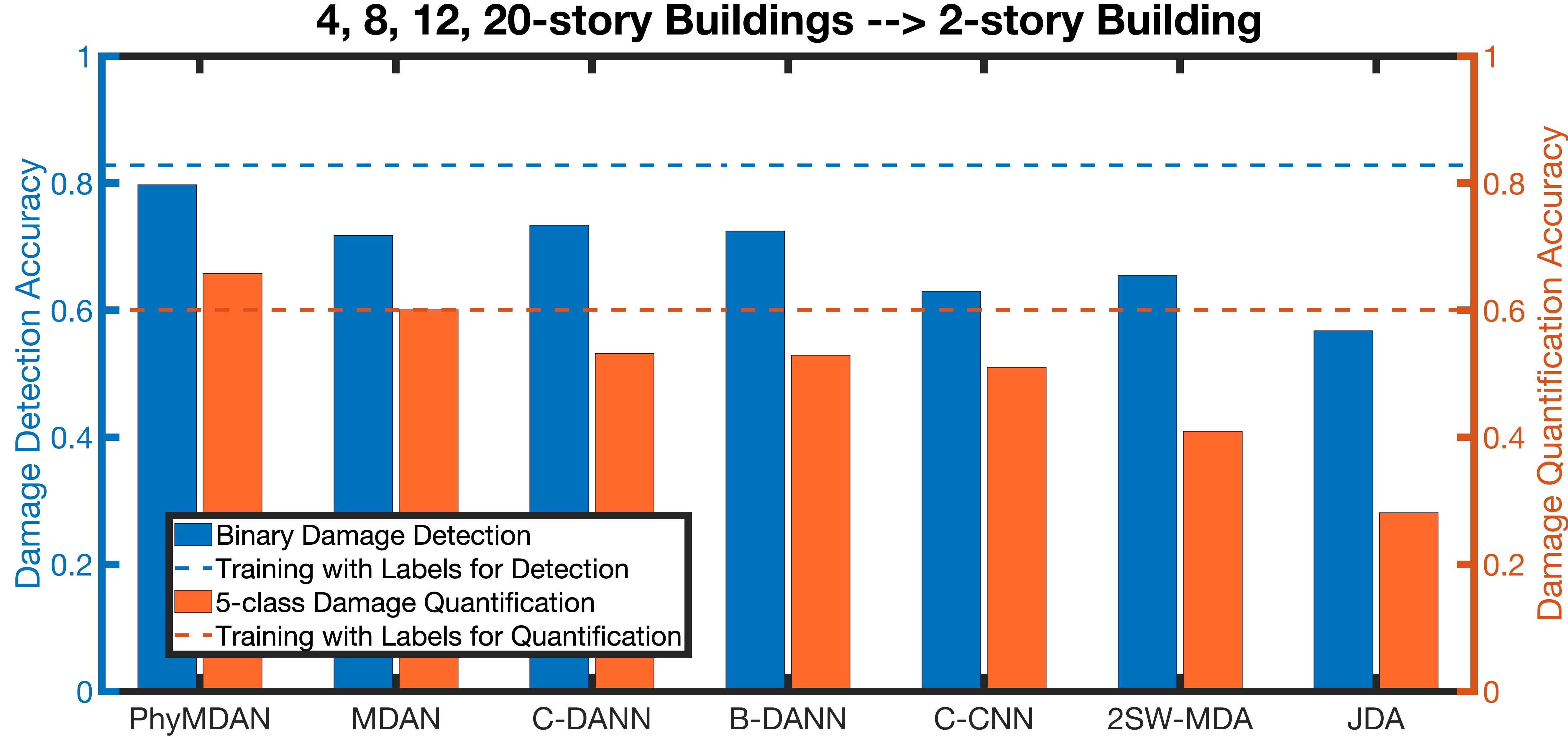}
		\caption{}
		\label{fig:2a}
	\end{subfigure}
	\begin{subfigure}{0.48\textwidth}
	\centering 
		\includegraphics[scale=0.0445]{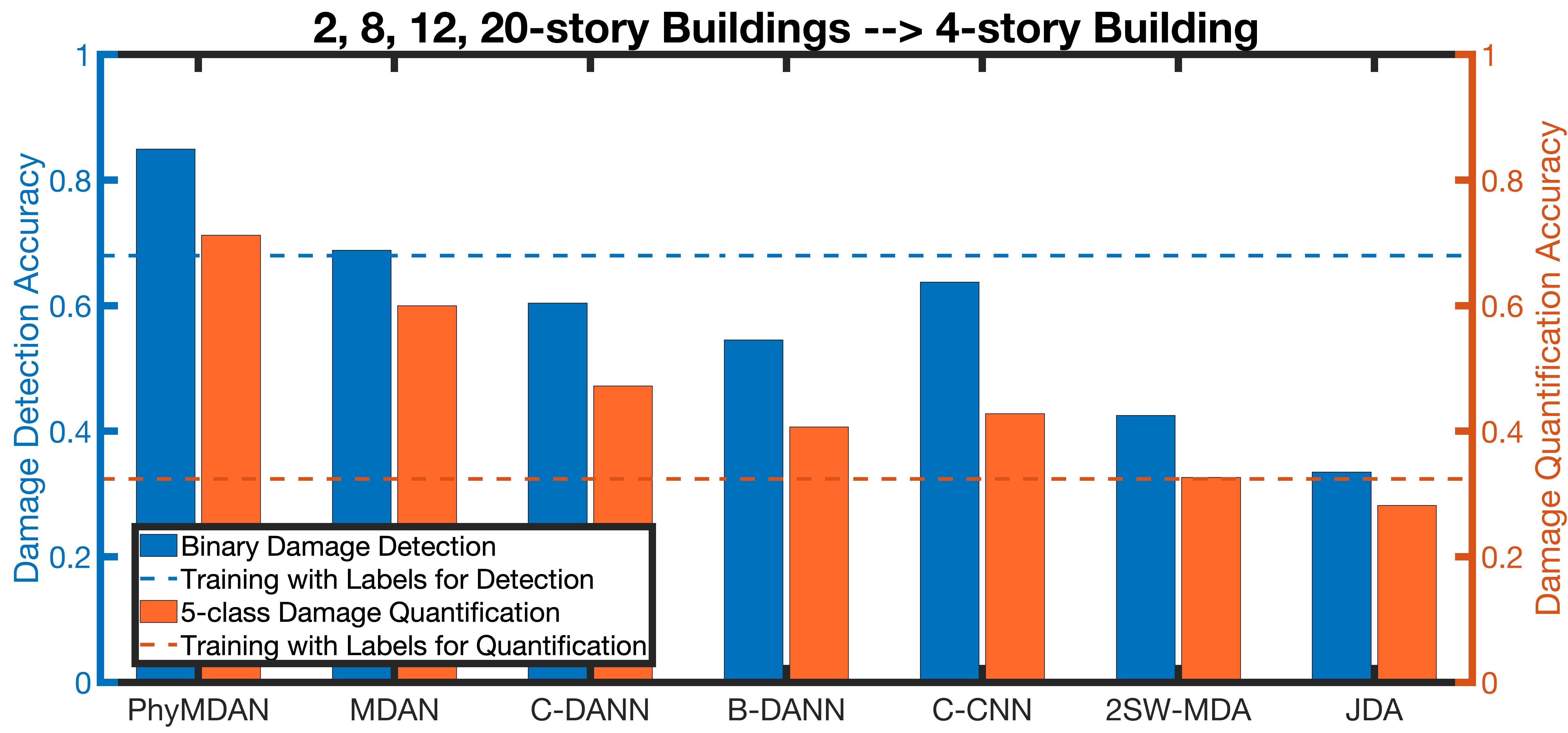}
		\caption{}
		\label{fig:2b}
	\end{subfigure}
	\begin{subfigure}{0.504\textwidth}
	\centering 
		\includegraphics[scale=0.045]{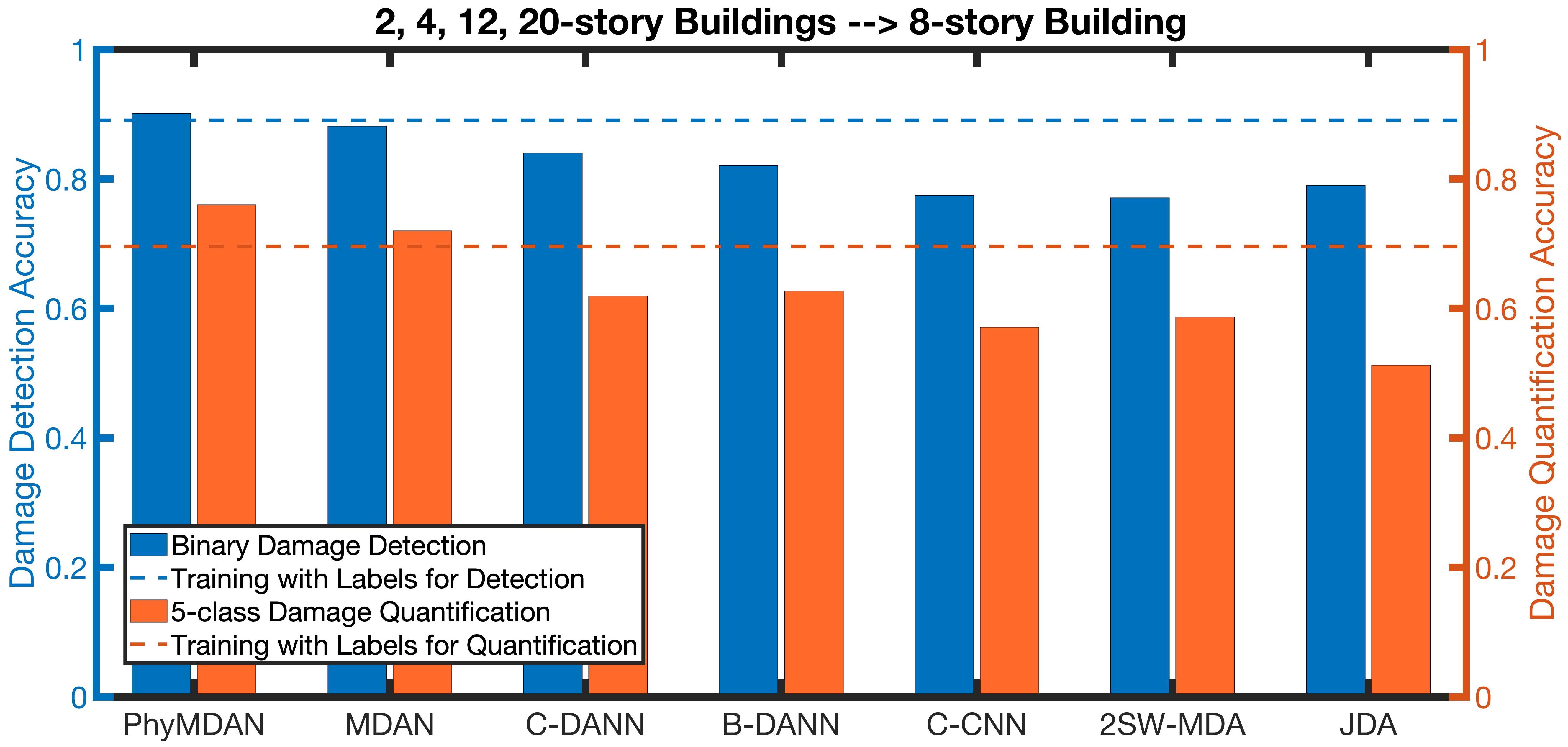}
		\caption{}
	\label{fig:2c}
	\end{subfigure}
	\begin{subfigure}{0.48\textwidth}
	\centering 
		\includegraphics[scale=0.055]{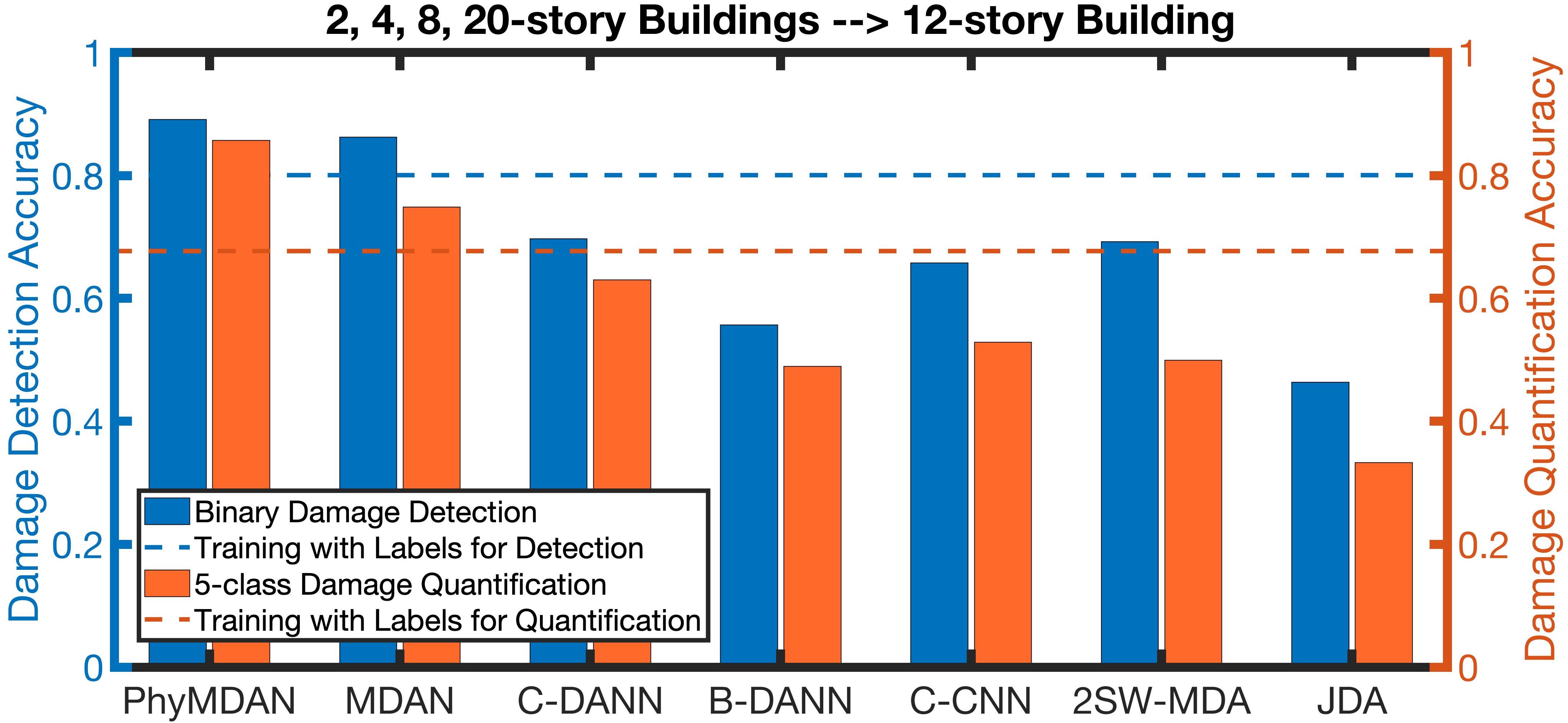}
		\caption{}
		\label{fig:2d}
	\end{subfigure}
	\begin{subfigure}{\textwidth}
	\centering 
		\includegraphics[scale=0.052]{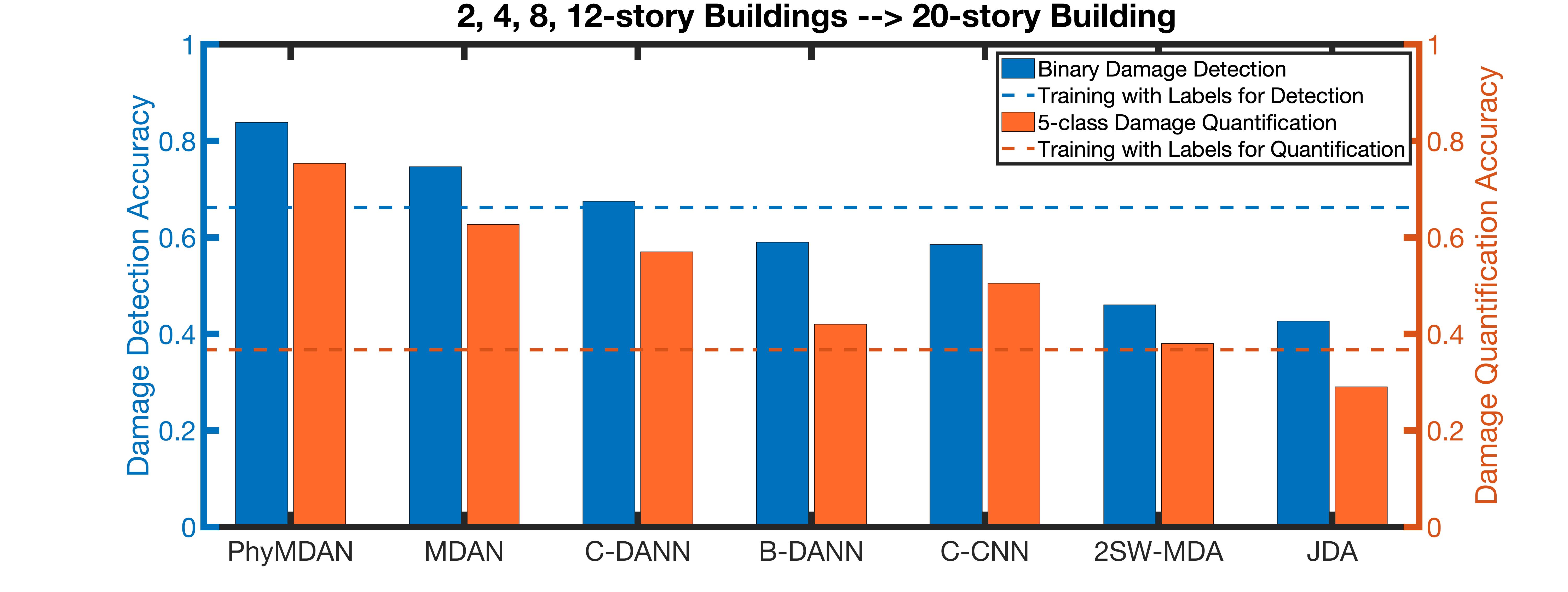}
		\caption{}
		\label{fig:2e}
	\end{subfigure}
	\caption{(a), (b), (c), (d), (e) compares the performance of knowledge transfer from other four buildings to the 2, 4, 8, 12, and 20-story building, separately, using our framework and other benchmark methods. The results include binary damage detection (blue) and 5-class damage quantification (orange). The dotted lines represent the damage prediction accuracy when directly training using the target domain labeled data as reference (ideal case). The results show that our framework outperform other methods.}
	\label{fig:result2}
\end{center}
\end{figure}
\subsection{Benchmark Methods}\label{sec:base}
\textcolor{black}{We compare the performance of our framework with seven other different methods, including Multi-source Domain Adversarial Networks (MDAN)~\cite{zhao2018adversarial}, Domain Adversarial Neural Network with Combined multiple source domains as one single source domain (C-DANN)~\cite{ganin2016domain}, Domain Adversarial Neural Network with the Best single source domain (B-DANN)~\cite{ganin2016domain}, Convolutional Neural Network with Combined multiple source domains as training set (C-CNN), Two-Stage Weighting Multi-source Domain Adaptation  (2SW-MDA)~\cite{sun2011two}, Joint Distribution Adaptation with the Best single source domain (JDA)~\cite{long2013transfer} and directly training with target labels}. MDAN is a multiple source adversarial domain adaptation method with uniform weight across different source domains. We use MDAN as one baseline method to show the model performance improvement due to our physics-guided loss design.  C-DANN combines the data from all source domains as single source domain dataset, and use single source domain adaptation method named as Domain Adversarial Neural Network (DANN)~\cite{ganin2016domain} to transfer knowledge from single source domain to the target domain. B-DANN transfers the knowledge from each source domain to the target domain using DANN~\cite{ganin2016domain}, and then selects the one with the best performance to report. C-DANN and B-DANN have the exactly same architectures of feature extractor and damage predictor with our framework. The only difference is that there is only one classifier in domain discriminator due to a single source domain. C-CNN directly uses deep convolutional neural network to train on the combined source data and predict on the target domain. The architecture for C-CNN is the combination of the architectures of the feature extractor and damage predictor in our framework mentioned in Section~\ref{sec:framework}. \textcolor{black}{Two-Stage Weighting Multi-source Domain Adaptation (2SW-MDA) is a classical multi-source domain adaptation method by reweighting the instances and source domains. Joint Distribution Adaptation (JDA) is a classical single-source domain adaptation aligning the joint distribution of source and target domains. When implementing JDA, we use support vector machine as the base classifier. We select one source domain to conduct the single-source-single-target domain adaptation first, and report the best performance across different source domains.} \textcolor{black}{We also evaluate the performance of directly training and testing on the target domain using deep convolutional neural network given ground-truth labels for target domain data (ideal case, supervised learning without using any source domain data)}. 

\vspace{-0.3cm}
\subsection{Knowledge Transfer Across Different Buildings on Simulation Data}\label{sec:simulate}
\begin{figure}[htp!]
\begin{center}
	\begin{subfigure}{0.49\textwidth}
		\includegraphics[scale=0.35]{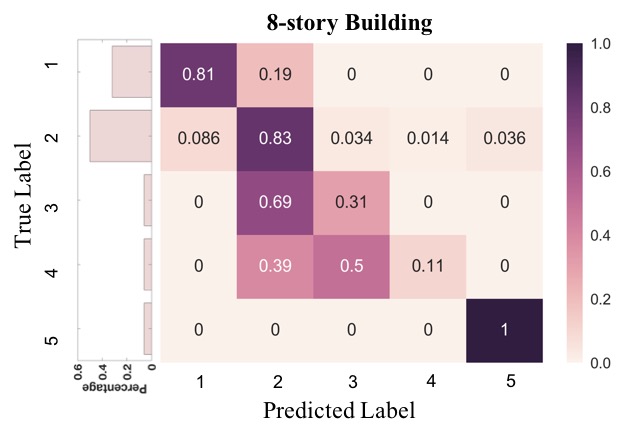}
		\caption{}
		\label{fig:3a}
	\end{subfigure}
	\begin{subfigure}{0.49\textwidth}
		\includegraphics[scale=0.35]{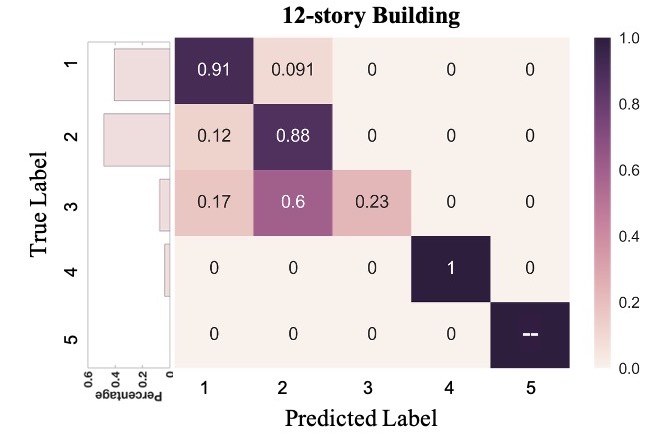}
		\caption{}
		\label{fig:3b}
	\end{subfigure}
	\caption{(a) Confusion matrix of the 5-class damage quantification result for knowledge transfer from other buildings to the 8-story building. (b) Confusion matrix of the 5-class damage quantification result for knowledge transfer from other buildings to the 12-story building. The left of each confusion map shows the density histogram of each damage class in the respective building's dataset.}
	\label{fig:result3}
\end{center}
\vspace{-0.3cm}
\end{figure}

We first evaluate the performance of our framework on knowledge transfer across different buildings using simulation data. Two different tasks, binary damage detection and $5$-class damage quantification, are conducted. 

To prepare the training and testing datasets for simulation data, we first conduct linear interpolation to align the timestamps of floor vibrations and impute the missing data for each building. %We then select the data collected under ground motions ranging from $0.18$ to $1.2$, which covers most moderate earthquakes as Figure~\ref{fig:pfa_sf} shows. 
We then conduct data augmentation as buildings have very limited structural response datasets. Based on the observations on the data, we take the sliding window length varying between $t-2.5$ seconds to $t$ seconds where $t$ is the time length of the raw vibration signal. We take the stride of sliding the window as $0.25$ seconds. %The vibration signals with too short ground motion duration are removed. 
Based on the selected data, we conduct Fast Fourier Transform (FFT) on the ground motion, the floor acceleration and the ceiling acceleration for each story-wise data sample. In our experiment, we take the the spectrum with frequency lower than $26$Hz, which corresponds to a $1000\times 1$ vector for each vibration signal. We choose lower frequency band since most fatal structural damages are induced by low-frequency seismic waves. The peak story drift ratio is simultaneously quantized into damage classes as ground-truth labels. We organized the data for each building by story level. \textcolor{black}{Before data augmentation, the $2,4,8,12,20$-story buildings have $241,195,382,283,154$ labeled data samples at each story, respectively. After data preprocessing and data augmentation, we finally obtain the datasets of the $2$, $4$, $8$, $12$, and $20$-story buildings containing $2600, 1400, 2500, 1750, 750$ labeled data samples at each story, respectively. }

We conduct separate experiments which transfer knowledge to each of $2$, $4$, $8$, $12$, and $20$-story building from the remaining buildings. \textcolor{black}{The final goal of this knowledge transfer is to classify story-level structural damage states}. \textcolor{black}{As for physics-informed weights, we pick one of the most commonly available physical properties, heights, as $U$ to calculate the weight distributions in Equation~\ref{eq:w2}. The justification of using weights is included in Section~\ref{sec:weight}}. As an example, we show the performance of our framework on transferring knowledge across the $2$nd floor of each building. We present the performance at the $2$nd floor since: 1) all the $5$ buildings have the $2$nd floor, and 2) as Figure~\ref{fig:psdr_sf} and Figure~\ref{fig:psdr_story} shows, the data distribution changes in lower floors are more complex than higher levels, which makes their analysis more challenging. 
\begin{figure}[htp!]
\begin{center}
	\centering
		\includegraphics[scale=0.06]{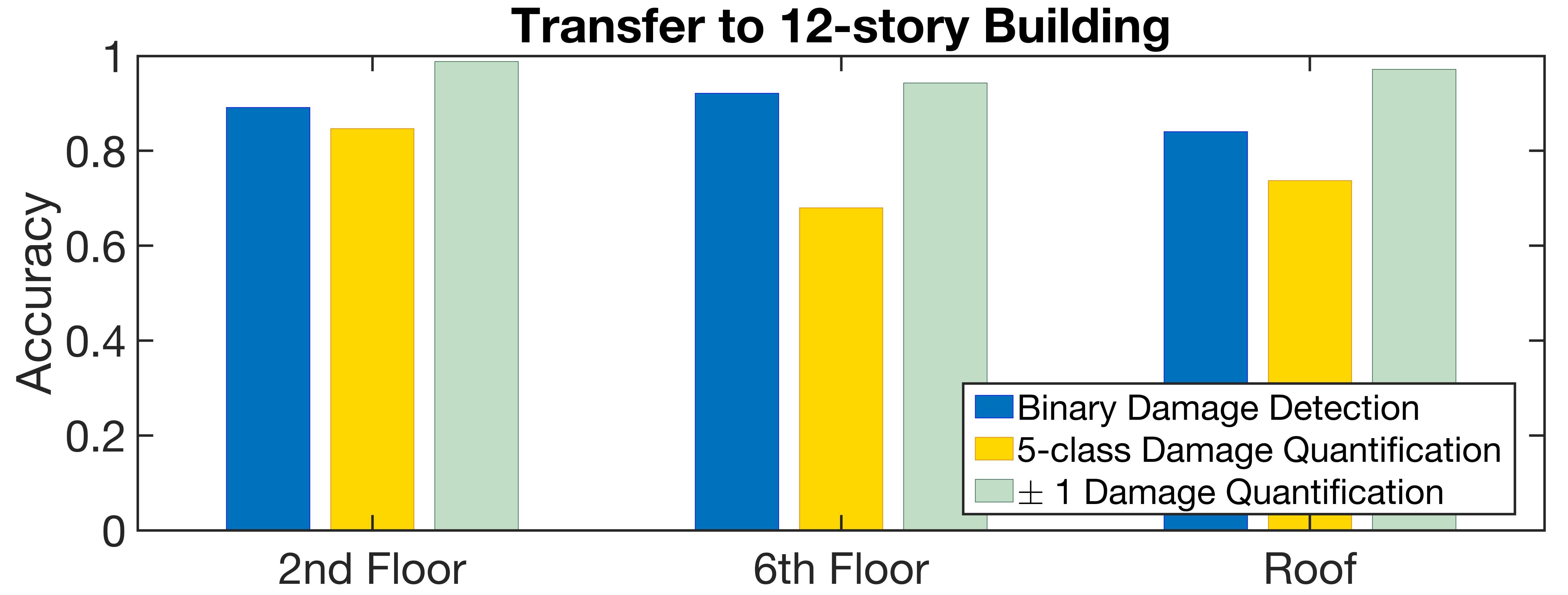}
		\caption{Performance of our framework to transfer knowledge from other buildings to the 12-story building across the 2nd story, 6th story and the roof story. The results on damage detection (blue), damage quantification (yellow) and $\pm 1$ damage quantification (green) are presented.}
		\label{fig:12_Story_floor}
\end{center}
\end{figure}

\textcolor{black}{The model training includes two phases: (1) finding the optimal hyperparameters, (2) training on the whole dataset. During the parameter tuning phase, we conduct 10-fold cross validation to find the best hyperparameters. All the source datasets are randomly split into ten folds. Hence, in each iteration of the cross validation, we have a source training set with size of 90\% and a source validation set with size of 10\%. The source validation set is used to validate the damage predictor and domain discriminator. Note that we do not need a testing set of source data, as our main goal is to predict the damage of the target building and source validation set is sufficient. The target data is also randomly split into 90\% training sets and 10\% validation sets. The target training set is used to train both damage predictor and domain discriminator. The target validation set is used to validate the domain discriminator. We do not need a testing set for domain discriminator as our main goal is to find the optimal feature extractor and damage predictor. After the key hyperparameter, $\lambda$, is selected, we input all source and target domain data to get the optimal feature extractor and damage predictor. Note that in all phases, we do not have any target labels available. Since we conduct unsupervised domain adaptation for the target data, training and test data are not separate for target data, just like other unsupervised learning approaches, and essentially all the target data is our test data since our goal is to find the labels for all the target data.}

We design different architectures for binary damage detection task and $5$-class damage quantification task. The architecture for damage quantification is shown in Table~\ref{table: Architecture} with $\sim 250$K parameters. Since damage detection is an easier task, we reduce the number of convolutional layers in feature extractor and damage predictor to $3$ and $2$, respectively. To evaluate the best performance, when training the model, we use Adam as optimizer~\cite{kingma2014adam} with the momentum of $0.9$ and the weight decay rate of $1e-4$. We take  $\lambda$ as $0.01\sim 0.05$ for damage detection 
and $\lambda$ as $0.2\sim 0.5$ for damage quantification. \textcolor{black}{$\lambda$ is selected via grid search and 10-fold cross-validation based on total loss function, damage prediction accuracy on source domain validation data, and domain discrimination accuracy on both source and target domain validation data. In empirical studies, it is found that optimal lambda often results in a stable convergence and a high damage prediction performance on source domain validation data no matter how the dataset is split.} The initial learning rate is set as $0.005$ for damage quantification and $0.0002$ for damage detection with learning decay rate of $0.1$.

Figure~\ref{fig:result2} shows the performance of our framework on knowledge transfer across 5 buildings for different damage diagnosis tasks. Figures~\ref{fig:2a}, \ref{fig:2b}, \ref{fig:2c}, \ref{fig:2d} and \ref{fig:2e} show the results of knowledge transfer to the $2$, $4$, $8$, $12$, $20$-story building from all the remaining buildings, respectively. Blue bars refer to the results of binary damage detection and orange ones represent the results of 5-class damage quantification tasks. We use the dotted line to represent the damage prediction accuracy of directly training on the target domain as reference (ideal case). Our framework achieves up to $42.77\%$ improvement on damage detection and $51.36\%$ improvement on damage quantification compared to the benchmark methods other than directly training on the target domain. Figure~\ref{fig:result3} presents the confusion matrix of damage quantification results for transferring knowledge from all the remaining buildings to the $8$-story and the $12$-story buildings. Since we do not have class 5 (collapse) for 12-story building data, we mark it as $1$ in the confusion matrix. \textcolor{black}{Figure~\ref{fig:result3} shows that the damage quantification accuracy achieves $76\%$ and $84.47\%$ for knowledge transfer from the 2,4,12,20-story buildings to the $8$-story building and from the 2,4,8,20-story buildings to the $12$-story buildings, respectively. Besides, the $\pm 1$ damage quantification accuracy achieves $97.4\%$ and $100\%$ for the $8$ and $12$-story buildings, respectively. In Figure 8, it could be found that the 12-story building does not have class 5, referring to the damage level of “collapse”. The ground motion used in simulation is collected from real-world earthquakes. We mainly utilized the data collected under service level, design level earthquake, and maximum considered earthquake, covering the most common earthquake excitations. Therefore, it is possible that some buildings did not collapse. Meanwhile, the most challenging part of post-earthquake structural damage diagnosis is distinguishing slight, moderate and severe damage states to help post-earthquake rescue and effective reconstruction. } Figure~\ref{fig:12_Story_floor} shows the performance of our framework to transfer knowledge from other buildings' $2$nd story, $6$th story and the roof story to the $2$nd story, $6$th story and the roof story of the 12-story building to compare the performance of knowledge transfer in different stories. The results on damage detection, damage quantification and $\pm 1$ damage quantification are presented.

 From Figures~\ref{fig:result2}, \ref{fig:result3} and \ref{fig:12_Story_floor}, we have 4 observations based on the results: 1) ``no damage" vs ``slight damage" data points, or ``moderate damage" vs ``severe damage" are difficult to classify. 2) Except transferring to the 2-story building, our method can achieve comparable performance or outperform directly training on the target domain given labeled target domain data. This is because that for some buildings with very few data, the information is too limited for directly training and testing on the target domain even if the label is given, which makes the model easily overfit and reduce the prediction performance.  This also shows that our framework is good at integrating the information from different source domains to improve the knowledge transfer. %For the 2-story building, the 2nd story is the roof story, which may result in a different damage pattern. 
3) \textcolor{black}{For most buildings, multiple source domain adaptation methods (Our PhyMDAN, MDAN) outperform single domain adaptation methods (C-DANN, B-DANN). Compared to baseline like C-DANN and B-DANN, multiple source domain adaptation methods can effectively integrate and transfer knowledge from multiple other buildings with complex distribution shifts. In particular, by embedding the prior physical knowledge, our framework achieves the most effective integration of knowledge from various source domains. Deep neural network architectures also help improve the performance with higher expressive power. So our PhyMDAN and MDAN perform better than 2SW-MDA, and B-DANN performs better than JDA. }  4) Small size of unlabeled target domain data would constrain the performance of knowledge transfer. \textcolor{black}{For example, the damage diagnosis accuracy of different algorithms is lower on the 4-story and the 20-story building compared to other buildings. One possible reason is that the sample size of target data is relatively smaller due to limited data availability. There are only 750 samples for each story of the 20-story building after data augmentation.} Having few inputs from the target domain makes it difficult for the feature extractor to sufficiently learn the underlying marginal distributions and conduct domain-invariant transformation.

\subsection{Knowledge Transfer Across Different Buildings From Simulation to Real-world Diagnosis}\label{sec:realdata}

\begin{figure}[ht!]
    \centering 
    \includegraphics*[scale=0.07]{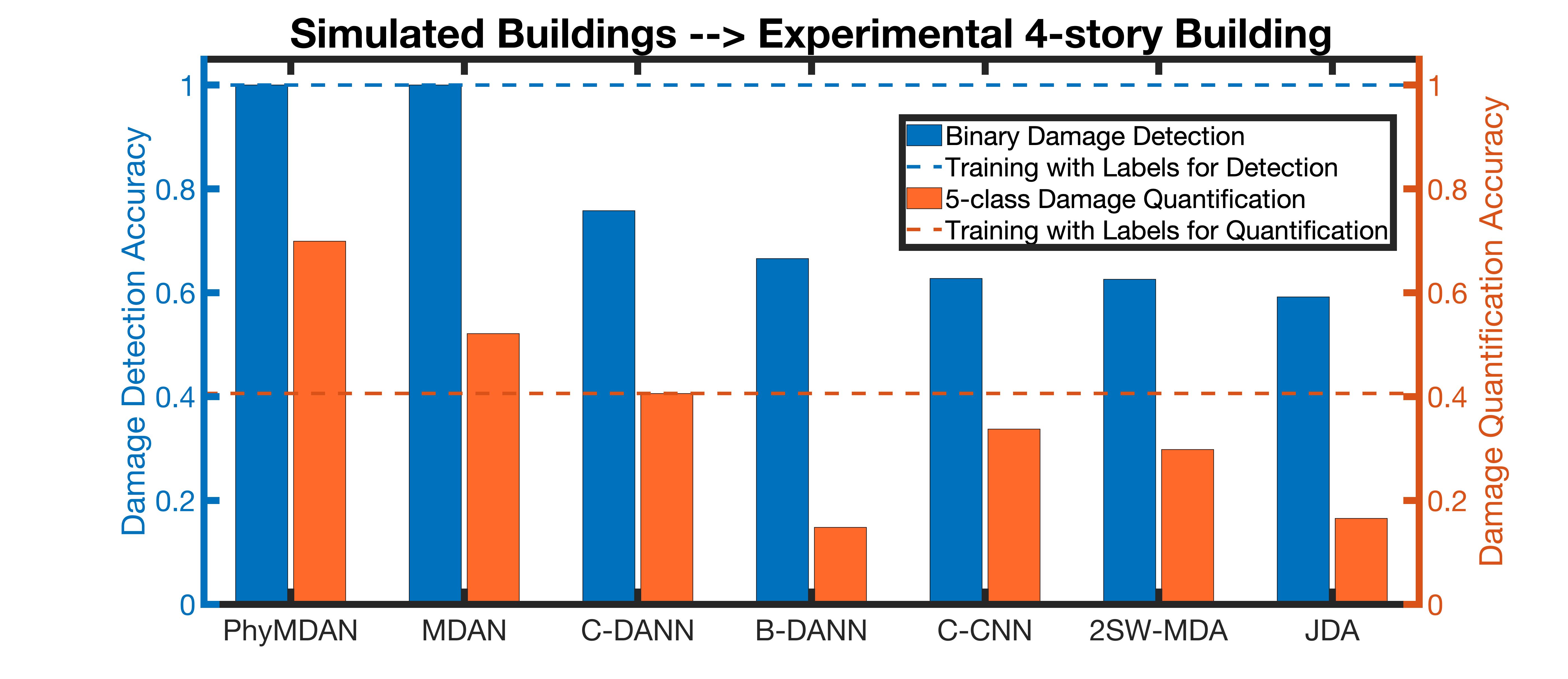}
    \caption{ This figure compares the domain adaptation performance between our method and other approaches on binary damage detection (blue) and 5-class damage quantification tasks (yellow) on a experimental 4-story building. We use the dotted line to represent the damage prediction accuracy of directly training on the target domain as reference.}
     \label{fig:real_Story}
\end{figure}

Compared to the simulated non-linear building model, experimental structures often have more complex non-linear load-deformation patterns. However, real-world seismic structural response data is often difficult to acquire. Here we validate the application to transfer the knowledge from the simulation data to the experimental structure for earthquake-induced structural damage diagnosis.

\begin{figure}[ht!]
    \centering 
    \includegraphics*[scale=0.4]{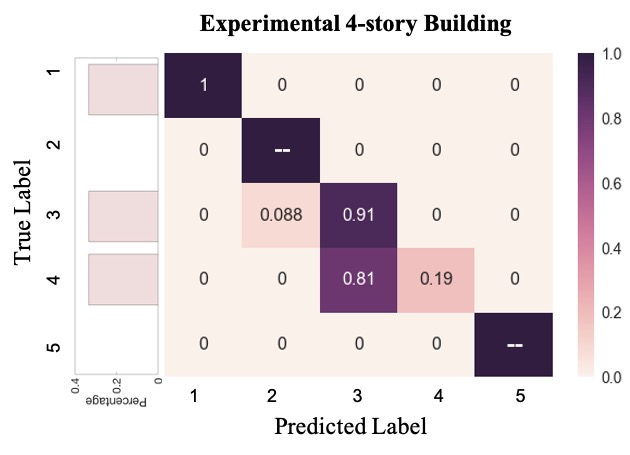}
    \caption{ Confusion matrix of the 5-class damage quantification result for knowledge transfer from simulation data to a experimental 4-story building. The left of the confusion map shows the density histogram of each damage class.} \label{fig:real_confusion}
\end{figure}

For the data preprocessing, we follow the same steps described for simulation data in Section~\ref{sec:simulate}. The difference is that we have more dense sliding window with small striding to augment the data, since the size of experimental building data is very limited. \textcolor{black}{Data augmentation results in $1500$ labeled data points for each story of the experimental building.} Note the damage class distribution is very imbalanced as Figure~\ref{fig:real_confusion} shows, which is common in experimental data.  We use the same architecture as Table~\ref{table: Architecture} for damage quantification. For damage detection, a model with smaller sizes of neural networks compared to the one for simulation data is employed. We use Adam with the momentum of $0.9$ and the weight decay rate of $1e-4$ as optimizer~\cite{kingma2014adam}. For damage quantification, we found the performance achieves the best when we set $\lambda =0.5\sim 1$. \textcolor{black}{The selection of $\lambda$ is the same with Section~\ref{sec:simulate}}.  We take $\lambda$ as $0.01\sim 0.05$ for damage detection. The initial learning rate is set as $0.005$ for damage quantification and $0.0002$ for damage detection with learning decay rate of $0.1$. The selection of hyper-parameters is based on our empirical study. \textcolor{black}{The partition of training and validation set is the same as the way in Section~\ref{sec:simulate}.}.

As a result, our framework can achieve $100\%$ damage detection accuracy and $69.93\%$ 5-class damage quantification accuracy, as shown in Figure~\ref{fig:real_Story}. In the task of damage detection, our framework has the same performance as MDAN, and outperforms other methods, which validates that multiple source domain adaptation methods have unique advantages on fusing the information from different source domains. On the more difficult task of damage quantification, our method outperforms all other methods, which shows the importance of discovering and combining the physical knowledge about the domains into the loss design.  Figure~\ref{fig:real_confusion} shows the confusion matrix for the damage quantification result, which indicates a $100\%$ $\pm 1$ damage quantification accuracy. Since we do not have class 2 (slight damage) and class 5 (collapse), we mark them as $1$ in the confusion matrix. Interestingly, the source dataset contains data points indicating slight damage and collapse, but the well-trained model only misclassifies a small group of data points belonging to moderate damage into slight damage, and avoids to classify any points into collapse.

\subsection{Characterizing the Performance of Different Physics-informed Weights}\label{sec:weight}
\textcolor{black}{We further explore the influences of different physics-informed weights on the final performance. The available quantified physical properties for simulation data includes static overstrength factor ($\Omega_s$), ductility factor ($\mu_T$), first-mode periods ($T_1$), height ($H$), as well as the number of stories ($N$). Note that in post-earthquake scenarios, $\Omega_s$ and $\mu_T$ are often unavailable in practice, but they are analyzed here to provide an understanding of how different physics-informed weights would impact the final knowledge transfer performance. }

\textcolor{black}{Table~\ref{tb:2} summarizes the physical characteristics of five different archetype steel frame buildings in simulation data, and more details of the parameter setting can be found in~\cite{elkady2015effect}. The overstrengh factor ($\mathbf{\Omega}_s$) refers to the ratio of the maximum base shear strength to the code-design base shear. The structural overstrength plays a vital role in protecting buildings from collapse, and is a result from many building physical properties such as material strength, confinement effect, member geometry and so on. The ductility factor ($\mu_t$) helps evaluate and quantify the global strength deterioration of the buildings. This factor is defined as a ratio of the global roof drift corresponding to a 20\% drop in the maximum base shear force to global yield roof drift~\cite{elkady2015effect}. The first-mode period is obtained based on eigenvalue analysis and is the fundamental property of the archetype buildings. More analysis on how these physical parameters reflect the structural dynamic responses to earthquakes can be found in~\cite{elkady2015effect, prestandard2000commentary, kircher2012tentative, ji2009effect,foutch2002modeling}.}

\textcolor{black}{
Given these building parameters, the physics-informed weights are further calculated based on Equation~\ref{eq:w1} and~\ref{eq:w2}. With Softmax normalization, the weights are obtained based on similarity of each physical parameter, separately. We also define combined weights by taking the average of normalized weights of individual physical parameters. As an example, we present the final weight distribution across four source domains when transferring knowledge from the $2, 4, 8,$ and $20$-story source buildings to the $12$-story building. In Figure~\ref{fig:12a}, it is shown that the heights, number of stories, first-mode periods, and combinations of heights and first-mode periods tend to weigh the $8$-story building the most, while based on the other physical properties or combinations of physical properties, the $20$-story building results in the highest weight. Meanwhile, as the height of building is proportional to the number of stories in our archetype buildings, the height-based weight distribution is almost the same with the weight distribution based on the number of stories. We only evaluate the combination of heights and other influence factors.   }

\begin{table*}
\centering
 \caption{Characteristics for archetype steel buildings in simulation data}\vspace{-0.2cm}
\label{tb:2}
\begin{tabular}{c|cccc}
\hline\hline
No. of stories & Static Overstrength&Ductility&First-mode Period&Height\\
$N$&$\mathbf{\Omega}_s$&$\mu_T$&$T_1$[s]&$H$[m]\\
\hline
 2&2.98&4.10&0.88&8.6\\
 4&1.75&4.60&1.51&16.6\\
 8&2.63&3.30&2.00&32.6\\
 12&2.09&2.70&2.70&48.6\\
 20&1.89&2.61&3.44&80.6\\
 \hline\hline
\end{tabular}
 \vspace{-0.3cm}
\end{table*}

\begin{figure}[h!]
\vspace{-0.5cm}
\begin{center}
	\begin{subfigure}{0.434\textwidth}
		\includegraphics[scale=0.0485]{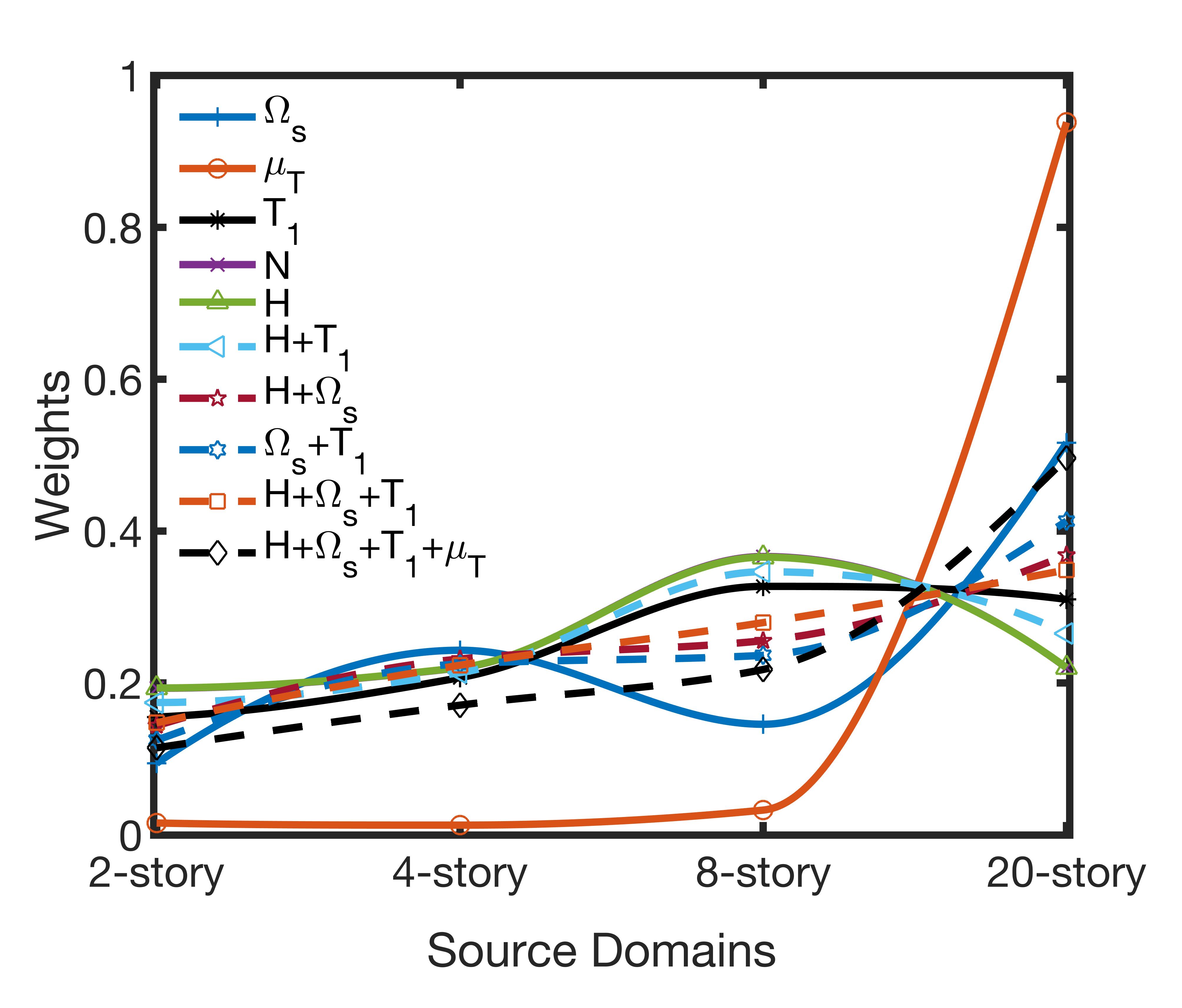}
		\caption{}
		\label{fig:12a}
	\end{subfigure}
	\begin{subfigure}{0.56\textwidth}
		\includegraphics[scale=0.07]{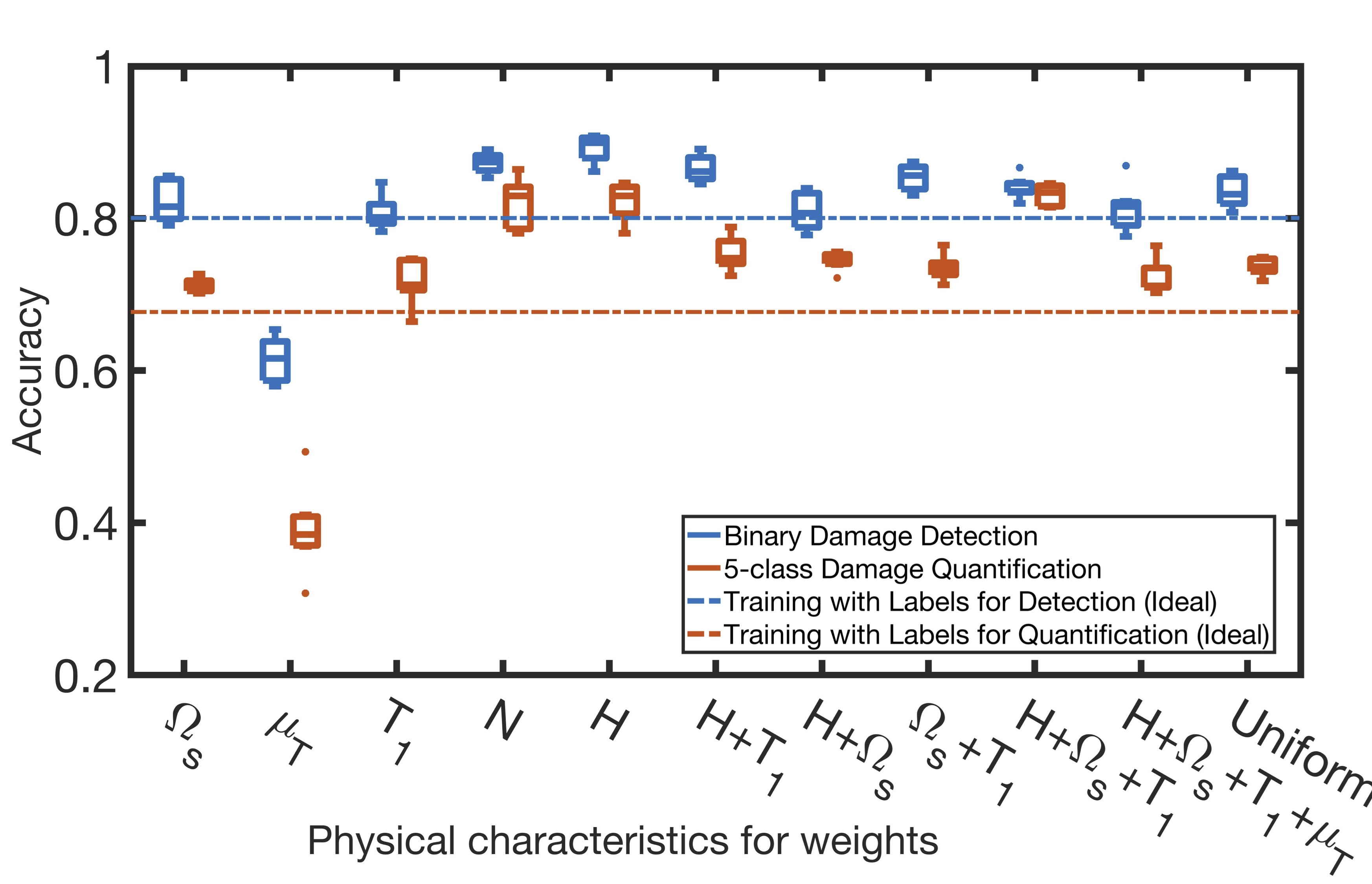}
		\caption{}
		\label{fig:12b}
	\end{subfigure}
	\caption{(a) Different physics-informed weights distribution across four source buildings calculated using different combinations of physical properties. (b) Accuracy of binary damage detection and 5-class damage quantification using different combinations of physical properties to calculate the physics-informed weights. The results are based on knowledge transfer from the 2,4,8,20-story buildings to the 12-story building in simulation data. ``Uniform" refers to the uniform weights, which is the baseline method MDAN. Blue bar refers to the results of damage detection. Orange bar refers to that of damage quantification. Blue and Orange dot line refer to the accuracy of directly training and testing on labeled target data.}
	\label{fig:weight_Selection}
\end{center}
\vspace{-0.5cm}
\end{figure}
\textcolor{black}{Based on different weights, we fine-tune the model and run the experiments multiple times to obtain accuracy of damage detection and quantification for transferring from the $2, 4, 8,$ and $20$-story source buildings to  the $20$-story building. Figure~\ref{fig:12b} shows the final performance using different physics-informed weights. PhyMDAN achieves the best performance by using the weights based on height ($H$), number of stories ($N$), and combinations of height, static overstrength factor  and first-mode periods ($H + \mathbf{\Omega}_s +T_1$), on both damage detection and damage quantification. Meanwhile, our PhyMDAN outperforms uniform weights (MDAN) when using the weights based on first-mode period, height, number of stories, combination of height and first mode period, combination of height and static overstrength factor, and combination of height, static overstrength factor, and first-mode period.  In the  weight distributions with best performance (based on $H$, $N$, and $H + \mathbf{\Omega}_s +T_1$), (1) the $2$-story source building always results in the lowest weight, indicating the least similarity with the $12$-story target building, (2) the $8$-story and $20$-story source buildings result in relatively higher weights as they share similar physical properties with the target building. However, the ductility factor-based weight has very low accuracy in both damage detection and quantification tasks. This is because ductility factor assigns a very high weight, $0.93$, to the $20$-story source building, making the process equivalent to single-source adversarial domain adaptation from the $20$-story building to the $12$-story building. As the $20$-story building has very few data samples ($750$ for each story), the final accuracy is lower than using other types of physics-informed weights that have less skewed distribution. Therefore, when selecting weights, we may need to avoid extremely unbalanced weight distributions. The results also suggest that the weights based on only ductility factors are not sufficient to estimate the similarities between different buildings.}

\textcolor{black}{In summary, the heights, number of stories, combination of heights and first mode periods, and combination of heights, static overstrength factor and first-mode periods always outperform than uniform weights and directly training and testing on the target building. This evidence shows that physics-informed weights provide reasonable hints to better integrate the information from multiple source domains and enhance the final knowledge transfer performance. }

\subsection{Characterizing the Training Process and Effect of $\lambda$}\label{sec:hyperparam}
\begin{table*}
\small
\centering
 \caption{$\lambda$ versus Damage Quantification Accuracy (2,4,8,20-story buildings $\rightarrow$ 12-story building)}
\label{tb:3}
\begin{tabular}{c|cccccc}
\hline\hline
$\lambda$ & 0.01&0.1&0.3&0.5&0.7&1\\
\hline
Domain Validation Acc (Target) &0.5886&0.8743&0.8343&0.7771&0.8171&1.00\\
 Damage Validation Acc (Source)&0.7834&0.8910&0.8013&
 0.8828&0.7655&0.8703\\
 Damage Test Acc (Target) &0.6177&0.6606&0.8240&\textbf{0.8447}&0.7903&0.6486\\
 \hline\hline
\end{tabular}
\end{table*}

\begin{figure}[ht!]
    \centering 
    \includegraphics*[scale=0.15]{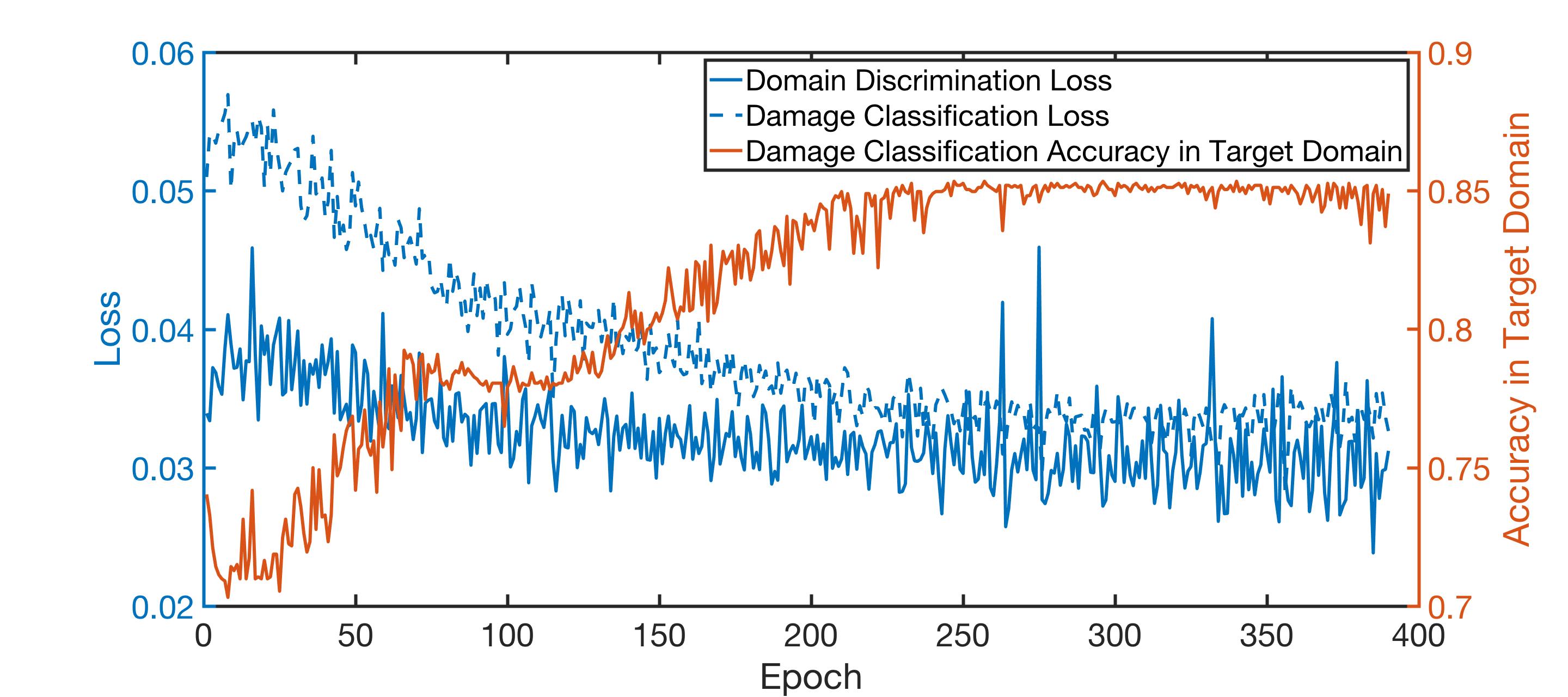}
    \caption{ The loss of domain discrimination loss, the loss of damage classification on the source domains, and the accuracy of damage classification on the target domain change with training epochs. }\label{fig:loss_visual}
\end{figure}
In this section, we characterize the training process, to provide insights and experiences about the optimization of our framework. A key problem for adversarial framework is how to find a stationary saddle point in minimax/maximin optimization. For our framework, this problem is more critical and difficult to resolve, because our framework combines three neural network architectures, which makes it more difficult to guarantee a stationary saddle point. Here we visualize the changes of the damage classification loss  $ \dfrac{1}{n}\sum_{i=1}^n w_i \pazocal{L}^i_M(\theta_e, \theta_m)$ and domain discrimination loss $ \dfrac{1}{n}\sum_{i=1}^n w_i \pazocal{L}_{D_i}(\theta_e, \theta_{d_i})$ in Figure~\ref{fig:loss_visual}. The figure shows that both loss keep fluctuating during the training epochs due to adversarial training scheme. The damage prediction accuracy on the target domain varies a lot at the early training stage, and finally converge a stable point. A basic insight to stabilize the training is that, in architecture design, the sub-classifiers in domain discriminator should be kept as light-weight nets but sufficiently powerful to conduct binary classification on the extracted features. However, sometimes it would be difficult to find a saddle point for the training framework. How to stabilize the training of adversarial frameworks needs to be resolved  in the future work. 
\begin{figure}[ht!]
    \centering 
    \includegraphics*[scale=0.3]{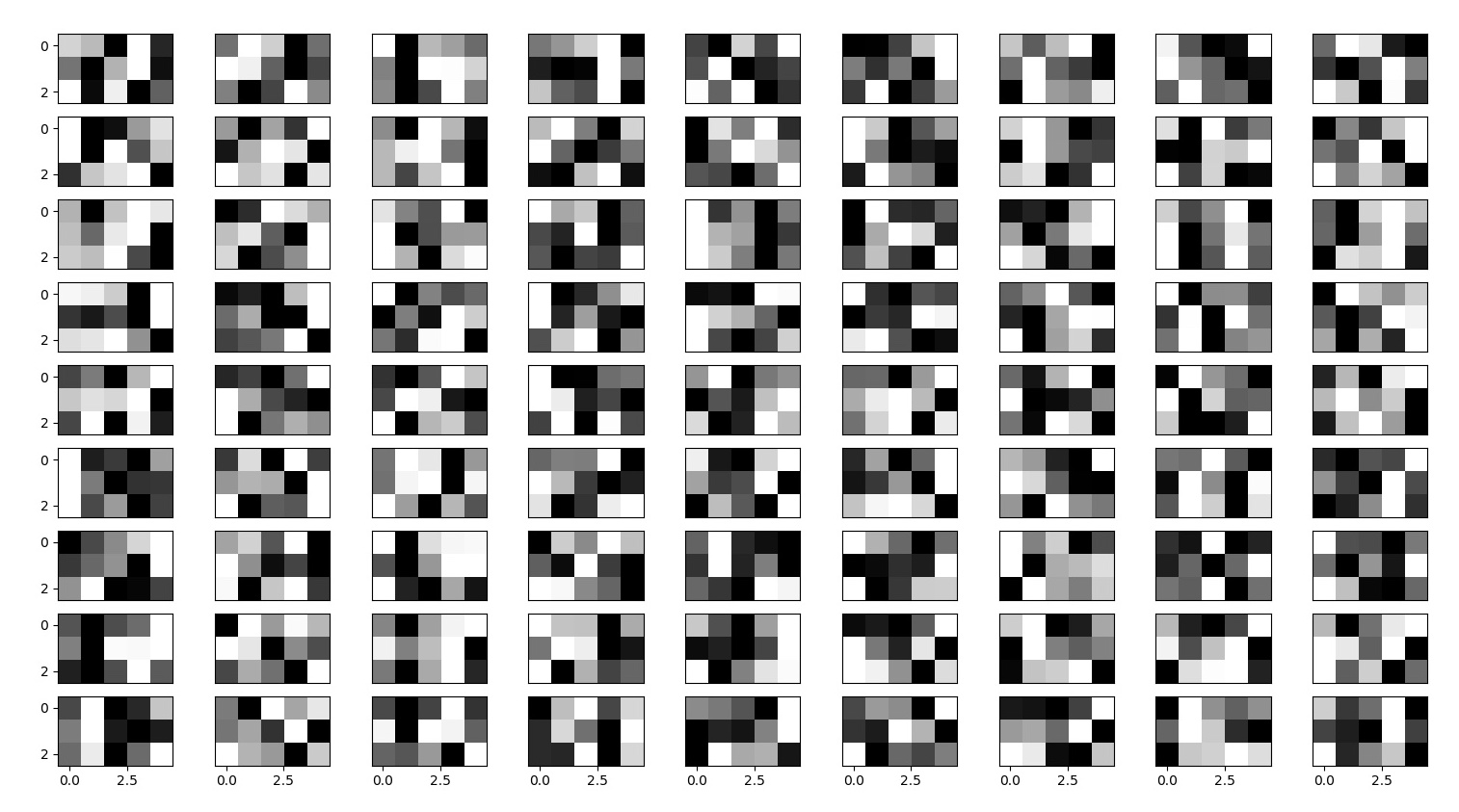}
    \caption{Visualized kernel for the first convolutional layer of the learned feature extractor described in the Table~\ref{table: Architecture}. \textcolor{black}{The feature extractor is learned for transferring knowledge from simulated buildings to the experimental building.}  There are 3 groups of 81 kernels with size of $5\times 1$. The 3 groups focus on extracting information from floor response frequency, ceiling response frequency, and ground motion frequency, which is ordered from top to bottom in the figure.}
    \label{fig:interpretable}
\end{figure}

The tuning of parameter $\lambda$ is another key for the network optimization. $\lambda$ controls the trade-off between the discriminativeness and domain-invariance of the extracted features. In our experiment, we found that different architectures have different optimal $\lambda$. If the architecture of damage predictor is more complex,  or if we have limited data or more noisy data, we need higher $\lambda$ to ensure the domain-invariance of the extracted features. \textcolor{black}{Parameter tuning for $\lambda$ has three rules of thumb: 1) the discriminator loss, predictor loss, and total loss need to converge, if it does not converge, we need to change $\lambda$ and tune the learning rate correspondingly; 2) training accuracy of domain discriminator need to converge to 100\%, 3) a good $\lambda$ can better balance domain discriminator and damage predictor. In Table 3, we present the domain validation accuracy on target validation set (175 samples), damage prediction validation accuracy on source validation set (725 samples), and damage prediction test accuracy on target test set (1750 samples). When the domain discrimination accuracy is relatively low and damage prediction accuracy on the source validation data is relatively high, the damage prediction accuracy on the target test data achieves the highest value. This means that an optimal trade-off is approximated -- the domain discriminator cannot distinguish the differences between the source and target domain accurately, and damage predictor can predict damage states based on extracted features accurately on source data. This trade-off agrees with our intuition to find a domain-invariant and damage-discriminative feature space. }

We also visualize the kernel in the first convolutional layer of learned feature extractor. In Figure~\ref{fig:interpretable} shows, it shows that there are 3 groups of 81 kernels with size of $5\times 1$. The 3 groups focus on extracting information from floor response frequency, ceiling response frequency, and ground motion frequency. We can find that most kernels are active, which means our network parameters are not redundant. It can be found that in some groups, the kernels focus on extracting information in the same frequency band, while in some other groups, the kernels for floor and ceiling responses tend to focus on opposite frequency band to the kernel for ground motion information.

\section{Conclusion}\label{sec:con}

In this work, we introduce a new modeling framework to adapt and transfer the knowledge learned from different buildings to diagnose the earthquake-induced structural damages of another building without any labeled data. Our adversarial domain adaptation approach extracts domain-invariant and damage-discriminative feature representations of data from different buildings. To the best of our knowledge, this is the first framework to address the multiple source domain adaptation challenge in post-earthquake building damage diagnosis without any labels of the target building. This end-to-end framework integrates data augmentation, input feature extraction, and domain adaptation for damage detection and damage quantification. Besides, this framework is flexible to combine as much available information in historical datasets from other different buildings as possible to help diagnose the current building, which ensures its application practicalities. In this framework, we design a new physics-guided loss function based on fuzzy physical knowledge about buildings to eliminate the uncertainties introduced by those source buildings with less physical similarities to the target building. We prove that this new loss provides a tighter upper bound the damage prediction risk on the target building. 

We evaluate our framework using both simulation data and experimental data, including $5$ different buildings under $40$ earthquakes for simulation and an experimental 4-story building subjected to incremental dynamic analysis. Our method achieves upto $90.13\%$ damage detection accuracy and $84.47\%$ damage quantification accuracy on simulation data. We also successfully transfer the knowledge learned from  simulation data to experimental data with $100\%$ damage detection accuracy and $69.93\%$ damage quantification accuracy, which outperforms other state-of-the-art benchmark  methods. The theoretical analysis and experimental results show the potential of our physics-guided adversarial domain adaptation framework to sufficiently distill and transfer the knowledge of seismic building damage patterns for building damage diagnosis, which is important for enabling fast and accurate post-earthquake building damage estimation. 

\newpage
\bibliographystyle{elsarticle-num}
\bibliography{reference}
\end{document}